\documentclass[aps,prx,reprint,superscriptaddress,nofootinbib,floatfix]{revtex4-2}

\usepackage{graphicx}
\usepackage{comment}
\usepackage{tabularx}
\usepackage{bm}
\usepackage{amsmath,amssymb,mathtools}
\usepackage{array}
\usepackage{booktabs}
\usepackage{hyperref}
\hypersetup{colorlinks=true,linkcolor=blue,citecolor=blue,urlcolor=blue}
\usepackage[T1]{fontenc}
\usepackage[protrusion=true,expansion=false]{microtype}

\newcommand{\Tr}{\operatorname{Tr}}

\makeatletter

\@ifundefined{lemma}{%
  \newtheorem{lemma}{Lemma}
}{}

\@ifundefined{proof}{%
  \newenvironment{proof}{%
    \par\addvspace{6pt}%
    \noindent\textit{Proof.}\ \ignorespaces
  }{%
    \unskip\nobreak\hfill\ensuremath{\square}%
    \par\addvspace{6pt}%
  }
}{}

\makeatother

\begin{document}

\title{Imaging Stars at the Quantum Compatibility Limit}

\author{Xinyao Guo}
\affiliation{Frontier Science Center for Quantum Information, Department of Physics, Tsinghua University, Beijing 100084, China}
\author{Haixing Miao}
\affiliation{Frontier Science Center for Quantum Information, Department of Physics, Tsinghua University, Beijing 100084, China}
\author{Zheng Cai}
\affiliation{Department of Astronomy, Tsinghua University, Beijing 100084, China}
\author{Huan Yang}
\email{hyangdoa@tsinghua.edu.cn}
\affiliation{Department of Astronomy, Tsinghua University, Beijing 100084, China}

\begin{abstract}
Imaging astrophysical sources with a multi-station interferometer is intrinsically a multiparameter quantum-estimation problem.  {Using tools from multiparameter quantum metrology,} we show that time-resolved repetitive or adaptive measurements in an \(N\)-station array suffer a fundamental array-level incompatibility among visibility estimators. Collective measurements,  {which coherently process the received starlight across multiple time bins in a single joint readout}, remove the array-size penalty up to an order-unity factor, yielding an asymptotic \(O(\sqrt{N})\) enhancement for the  {directional-averaged} SNR of visibility measurement. 
 {We propose a memory-assisted interferometric architecture
designed to implement collective readout through coherent storage and joint quantum processing.} Imaging simulations and Fisher-information analyses demonstrate that collective measurements improve image reconstruction in near-term arrays and enhance the resolving power of future long-baseline architectures, with pronounced benefits for representative AGN targets such as NGC~4151 and 3C~273. These results highlight collective measurement as a promising building block for future quantum-assisted interferometric arrays for stellar imaging.
\end{abstract}

\maketitle

\section{Introduction}
Long-baseline interferometric arrays that directly interfere the collected light is a promising route to high-resolution imaging of astrophysical sources.
Optical and near-infrared arrays with kilometer-scale baselines would approach the microarcsecond regime, opening direct access to visible-band imaging of scientifically compelling targets, including active galactic nuclei, broad-line regions, and jet-launching environments
\cite{BaldwinHaniff2002,Monnier2003,Monnier2007,Petrov2007,Tatulli2007,Benisty2009,Gravity2018}.
This possibility has renewed interest in interferometric telescope arrays beyond existing facilities. Multiple architectures have been proposed linking distant stations into large-scale interferometric networks in both classical and quantum settings~\cite{Chelli2009,Lacour2014,Ceus2011,Ceus2013,Tsang2011,Gottesman2012,
Khabiboulline2019PRL,Khabiboulline2019PRA,Huang2022,Marchese2023,Brown2023,Wang2025,Huang2024,
Stas2026}.  

 {Interferometric imaging with a multi-station array is intrinsically a multi-parameter quantum-estimation problem: the spatial brightness distribution of an astrophysical source is encoded in the amplitudes and phases of the inter-station coherence of the received optical field. Each pair of stations, or baseline, probes a distinct Fourier component of the source structure, so image formation requires the inference of these distributed coherence parameters ~\cite{BaldwinHaniff2002,Monnier2003,Tsang2011}. Existing network proposals span various classical and quantum platforms
\cite{Chelli2009,Lacour2014,Ceus2011,Ceus2013,Tsang2011,Gottesman2012,
Khabiboulline2019PRL,Khabiboulline2019PRA,Huang2022,Marchese2023,Brown2023,
Huang2024,Wang2025,Stas2026,Wang2026MemoryAssistedInterferometer}.
These designs are primarily guided by the feasibility of implementation and are therefore often suboptimal in terms of multiparameter estimate sensitivity}.
However, the optimal  {networking strategy and measurement precision}  in a non-local quantum network that contains many stations have
not been systematically discussed.


 {In this work, we address this optimal networking problem in three steps. First, by formulating nonlocal imaging as a multiparameter-estimation problem and invoking quantum-metrological precision limits, we identify collective measurement as the key to quantum-information-preserving networking \, \cite{GillMassar2000,Ragy2016,Yamagata2013,Demkowicz2020,Chen2022Hierarchy,Chen2022Incompatibility,Imai2026Hierarchy}. Specifically, in an \(N\)-station array, any time-resolved repetitive or adaptive single-copy strategy suffers from an array-level incompatibility among visibility estimators. Collective readout removes the array-size-dependent penalty up to an order-unity factor, yielding an asymptotic \(O(\sqrt{N})\) enhancement in the direction-averaged visibility SNR. Second, guided by this result, we propose a memory-assisted network architecture designed to implement such collective measurements. Finally, through imaging simulations and Fisher-information analysis, we show that the sensitivity gain from collective measurement directly yields improved reconstruction quality and an enhanced ability to resolve representative astrophysical sources, such as NGC 4151 and 3C 273, even in the presence of detector loss and atmospheric piston noise, for both practically feasible near-term arrays and more futuristic long-baseline architectures.  Together, these results establish collective readout as a scalable route toward stellar imaging at the quantum-compatibility limit.}

\section{Interferometric imaging as a measurement problem} 
We consider an interferometric network comprising \(N\) individual stations that coherently receive weak thermal light from a remote astrophysical source. Figure~\ref{fig:illustration} illustrates how such an array enables high-resolution imaging.

\begin{figure}[t]
    \centering
    \includegraphics[width=\columnwidth]
    {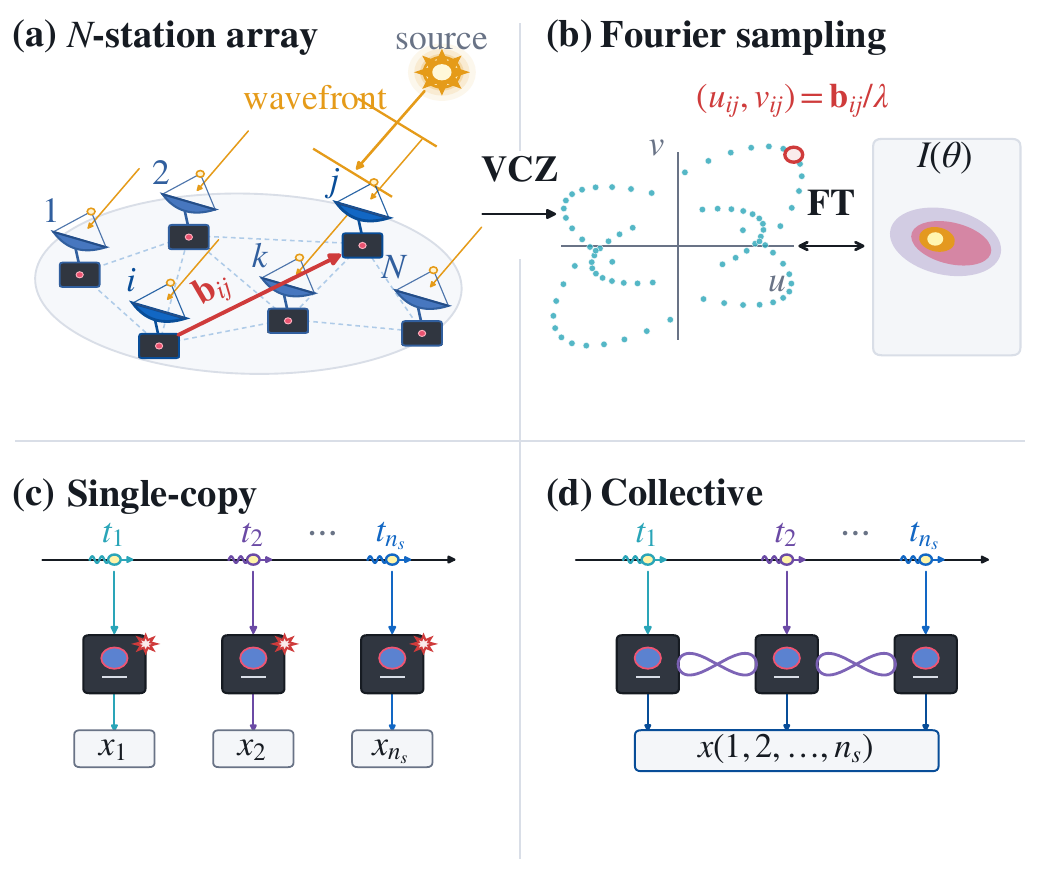}
    \caption{Schematic of a nonlocal interferometric array for high-resolution imaging.
\textit{Top:} Physical picture of interferometric imaging, in which multiple
baselines sample spatial-frequency components that jointly reconstruct the
source. \textit{Bottom:} Comparison of single-copy and collective measurements:
for the former, quantum-state collapse happens for each copy with its time label, whereas the latter
coherently processes multiple copies before a joint, time-unresolved collapse.}
    \label{fig:illustration}
\end{figure}

From a quantum perspective, in the weak-light regime and in the absence of instrumental noise, the quantum state of the \(N\)-site optical field can be consistently truncated to its vacuum and one-photon sectors~\cite{Tsang2011,Khabiboulline2019PRA,Marchese2023}:
\begin{equation}
\rho\approx(1-s)\rho^{(0)}+s\rho^{(1)}\,,
\end{equation}
where $s\ll1$ is the average occupation number, $\rho^{(0)}$ denotes vacuum state in all mode, and $\rho^{(1)}$ is the subspace with single photon occupation.  {Multi-photon event in a coherence time of temporal mode only happens with a low probability of $O(s^2)$, hence, truncating the state to the vacuum and one-photon sectors introduces only a second-order correction, and doesn't modify the fundamental physics.}

According to the van Cittert–Zernike theorem, the normalized off-diagonal elements of the single-photon occupation part of the density matrix,  also known as the \textit{complex visibility}, is one-to-one mapped to the Fourier component of targeting astrophysical source via~\cite{Monnier2003,Tsang2011,CZ}
\begin{equation}
\begin{aligned}
\frac{\rho^{(1)}_{ij}}{\sqrt{\rho^{(1)}_{ii}\rho^{(1)}_{jj}}}
=\frac{
\int I_{\lambda}(\boldsymbol{\theta})\exp\!\left[-2\pi i\,\mathbf{b}_{ij}(t)\cdot\boldsymbol{\theta}/\lambda\right]d\theta_xd\theta_y
}{
\int I_{\lambda}(\boldsymbol{\theta})d\theta_xd\theta_y
}\,,
\end{aligned}
\end{equation}
where $I_{\lambda}(\boldsymbol{\theta})$ is the intensity distribution of source, $\mathbf{b}_{ij}$ is the position vector pointing from station $i$ to station $j$, and the integration runs over the field of view of local telescopes. Imaging stars therefore corresponds to a quantum measurement problem of obtaining a joint estimation of all off-diagonal elements of density matrix.

Another feature of interferometric imaging lies in the integrated observation for stellar interferometry. The integration time is typically much longer than the lifetime of a temporal mode, yet much shorter than the characteristic astrophysical timescale.  {The continuously received field can therefore be viewed as a stream of statistically independent and quantumly indistinguished copies of the same source-dependent quantum state, each encoding the same set of unknown classical parameters characterizing the astrophysical source.} How to operate these copies naturally defines two classes of readout~\cite{GillMassar2000,SuzukiYangHayashi2020,
DemkowiczGoreckiGuta2020,Conlon2023Collective,Imai2026Hierarchy}:
\begin{enumerate}
\item \emph{Single-copy measurement},  {including repetitive and adaptive measurement, in which each copy is measured and undergoes quantum-state collapse separately, yielding a time-labeled classical outcome. Information from different copies is combined only through classical
statistical averaging or feedback that updates
the working point of the subsequent measurement.}
\item \emph{Collective measurement},  {in which copies are preserved, coherently
combined, and then collapsed altogether by
one joint measurement. Such a joint readout allows different multi-copy probability amplitudes to interfere, and can access collective observables and outcome correlations unavailable to any single-copy strategy.}
\end{enumerate}

In the following part, we focus on the phase-estimation half of problem, that is, jointly read out all the visibility phases in the density matrix. We show that, collective measurement brings fundamental sensitivity enhancement in the multi-parameter imaging task.

\section{The curse of incompatibility}
Considering the case that we get sufficient copies of the targeting N-station thermal state. For a given multi-parameter measurement captured by a positive-operator-valued measure (POVM) $\boldsymbol{\Pi}$, its capability for 
joint parameter estimation is described by the Classical Fisher Information Matrix (CFIM), denoting by $J(\boldsymbol{\Pi})$\,\cite{Helstrom1976,BraunsteinCaves1994,Rao1945}. For an estimable linear combination of visibility phases represented by
\(
\hat \phi=\boldsymbol{c}^{\mathsf T}\boldsymbol{\phi}\), the scalar Fisher information and
a lower bound for minimum estimation variance
under $n$ individual measurements are 
\begin{equation}
F=\left[\bm c^{\mathsf T}J(\bf \Pi)^{+}\bm c\right]^{-1}\,,\quad
(\sigma_{\phi})^2\geq\frac{1}{n F}\,.
\label{eq:loop_phase_variance}
\end{equation}

We compare the performance of  single-copy measurement and collective measurement by comparing their CFIM and loop-wise Fisher information per average copy.
 {For single-copy measurement, the most natural realization is a uniform
edge-first measurement, in which optical mode at each station is divided equally
among its \(N-1\) baselines, and each matched pair of branches is interfered on
a balanced beam splitter, thereby projecting the photon uniformly onto
pairwise spatial-mode bases.  The output-port clicks then provide an
independent joint interferometric measurement of the visibilities. Notably, this measurement is also the blueprint of many quantum-assisted non-local networks, including the \(W\)-state-assisted implementations of the standard Gottesman–Jennewein–Croke (GJC) and Khabiboulline–Borregaard–de Greve–Lukin (KBGL) scheme\,\cite{Gottesman2012,Khabiboulline2019PRL,Khabiboulline2019PRA}. Mathematically, this measurement can be described by the following POVM:}
\begin{equation}
\Pi_{ij,\pm}^{\mathrm{edge}}
=
\frac{1}{N-1}
\left|\pm_{\vartheta}^{(ij)}\right\rangle
\left\langle\pm_{\vartheta}^{(ij)}\right|,
\quad
\sum_{i<j,\,\pm}\Pi_{ij,\pm}^{\mathrm{edge}}
=
I_N\,.
\end{equation}
Here $|\pm_{\vartheta}^{(ij)}\rangle
\equiv (|i\rangle\pm e^{-i\vartheta}|j\rangle)/{\sqrt{2}}$, and the basis $|i\rangle$ represents the superposition that the i-th station has single-photon occupation while the rest have no photon, and the working point $\vartheta_{ij}$ is set as $\arg(\rho_{ij})+\pi/2$ for phase measurement. 

The upper limit of CFIM is the Quantum Fisher information matrix (QFIM) $J^{Q}$\,\cite{BraunsteinCaves1994,Sidhu2021}, which is the single-parameter envelope of CFIM for all attainable measurements. Normalized by the QFIM, we could see that the time-resolved uniform pairwise measurement only saturates the QFI by:

\begin{equation}
\frac{F(\{\Pi_{ij}^{edge}\})}{F ^Q}=\frac{\alpha_ {}}{N-1}\,,
\end{equation}
where $\alpha {}\in [0,1]$ is a structural penalty caused by only adopting pairwise measurement\,\cite{Zhang2026ArraySPADE}. Here, the $N-1$ penalty arises from the incompatible nature within the  estimator family. Specifically, for repetitive or adaptive measurement on individual copies, two parameters can be jointly extracted at the maximum precision only if their optimal estimators---symmetric logarithmic derivative (SLD) operators---commute at the operator level~\cite{Ragy2016,Yang2019}. However, in interferometric imaging, SLD operators for estimating visibility phases $\theta_{ij} \equiv \arg \rho_{ij}$: $\hat{L}_{ij}={\rm i}(e^{{\rm i}\vartheta_{ij}}|i\rangle\langle j|-e^{-{\rm i}\vartheta_{ij}}|j\rangle\langle i|) $ , share a baseline station exhibit intrinsic non-commutativity
\begin{equation}\label{eq:commute}
\left[\hat{L}_{ij},\hat{L}_{ik}\right]={\rm i}(e^{{\rm i}(\vartheta_{ij}-\vartheta_{ik})}|i\rangle\langle j|-e^{-{\rm i}(\vartheta_{ij}-\vartheta_{ik})}|j\rangle\langle i|)
\neq 0\,.
\end{equation}

 {By applying the Gill–Massar one-copy information capacity bound to the imaging task~\cite{GillMassar2000}, we illustrate that} this curse of incompatibility is universal for any possible single-copy measurement. Specifically, any repetitive or adaptive POVM on single copy suffers an \(O(N)\) QFIM-averaged penalty relative to the single-parameter optimum, given by:
\begin{equation} 
\begin{aligned}
\left\langle \frac{F{(\Pi)}}{F^Q} \right\rangle_Q &\equiv \int_{S^{r-1}} d\boldsymbol{u} \frac{\boldsymbol{c}_ \top J_Q^{-1} \boldsymbol{c}}{\boldsymbol{c}^\top [J{(\Pi)}]^{-1} \boldsymbol{c}} \\
&\le \min\left\{1, \frac{N-1}{r}\right\} \sim O(N^{-1})\,,
\end{aligned}
\end{equation}
where $r$ is the rank of the QFIM, denoting the number of practically  {estimable parameters of physical interest}. In the non-degenerate scenario where $r = E = N(N-1)/2$, this expression collapses to a robust bound of $2/N$. Here, $\boldsymbol{u}_ {} = J_Q^{-1/2} \boldsymbol{c}_ {} / \sqrt{\boldsymbol{c}_ {}^\top J_Q^{-1} \boldsymbol{c}_ {}}$ represents a unit vector in the parameter space, with the integration taken uniformly over the unit sphere $S^{r-1}$. A simplified proof is provided in Appendix\,\ref{sec:bounds}.

Things become different when collective measurement that jointly operates on $n_s$ independent copies are allowed. In the asymptotic limit of $n_s\to \infty$, the condition for joint saturation of two directional QFI can be weakened from the operator-level commutativity to an ensemble-averaged \textit{weak commutativity}~\cite{Ragy2016,Yamagata2013,Demkowicz2020},
\begin{equation}\label{eq:weak_com}
[\hat{L}_{ij}, \hat{L}_{kl}]=0\to {\rm Tr}(\rho [\hat{L}_{ij}, \hat{L}_{kl}])=0\,,
\end{equation}

 {This weaker condition substantially enlarges the set of parameter directions
that can be estimated jointly at their single-parameter limits. A concrete four-station example can illustrate this distinction. Consider the
visibility-phase directions on edges \(12\), \(13\), and \(14\), with
corresponding SLDs \(\hat L_{12}\), \(\hat L_{13}\), and \(\hat L_{14}\).
For a direction
\(\hat L_v=\alpha\hat L_{13}+\beta\hat L_{14}\), a single-copy receiver
would require
\begin{equation}
[\hat L_{12},\hat L_v]
=
\alpha[\hat L_{12},\hat L_{13}]
+\beta[\hat L_{12},\hat L_{14}]
=0 .
\end{equation}
At a generic full-rank working point, the two commutators on the right-hand
side are nonzero and linearly independent operators. Consequently, no
nontrivial choice of \((\alpha,\beta)\) can make the two directions strongly
commuting.}

 {For collective measurement, however, define the real weak-commutator coefficients
\begin{equation}
b_{ef}\equiv
\frac{1}{2{\rm i}}\,
{\rm Tr}\!\left(\rho[\hat L_e,\hat L_f]\right).
\end{equation}
The state-dependent direction
\begin{equation}
\hat L_v\propto
b_{12,14}\hat L_{13}-b_{12,13}\hat L_{14}
\end{equation}
automatically satisfies
\begin{equation}
{\rm Tr}\!\left(\rho[\hat L_{12},\hat L_v]\right)=0 ,
\end{equation}
even though \([\hat L_{12},\hat L_v]\neq0\) generally remains true as an
operator identity. For example, when
\(b_{12,13}=b_{12,14}\neq0\), the weakly compatible direction is simply
\(\hat L_v\propto\hat L_{13}-\hat L_{14}\). Thus an asymptotic collective
receiver can attain the directional QFI for edge \(12\) and this particular
\(13\)--\(14\) combination simultaneously. The same block of copies can then
contribute to both estimates, rather than being divided between incompatible
single-copy settings, thereby improving the joint sensitivity.}

 {We generalize the insight gaining from this example to more general case.} We find that, at the asymptotic limit, incorporating collective measurement guarantees a superior scaling for the multi-parameter estimation performance. Specifically, for any given single-copy POVM on the single-copy, there is a series of ways to promote it as a collective POVM jointly acting on all the copies. And, the weighted-average performance gap between the single-copy measurement and the optimal promoted joint measurement is bounded by:
\begin{equation}
\begin{aligned}
\left\langle
\frac{F_{ \infty}^{\mathrm{col}}(\Pi)}
     {F {}{(\Pi)}}
\right\rangle_{A_{\Pi}} \equiv \int_{S^{r^\prime-1}} d\boldsymbol{u}_ {}^\prime &\frac{\boldsymbol{c}_ {}^\top [J(\Pi)]^{-1} \boldsymbol{c}_ {}}{\boldsymbol{c}_ {}^\top [J^{\rm col}{(\Pi)}]^{-1} \boldsymbol{c}_ {}}
\geq \frac{r^\prime}{2(N-1)},\\
\qquad
\left\langle
\frac{\mathrm{SNR}_{ \infty}^{\mathrm{col}}(\Pi)}
     {\mathrm{SNR} {}{(\Pi)}}
\right\rangle_{A_{\Pi}}
&\geq\sqrt\frac{{r^\prime}}{2(N-1)}\,.
\end{aligned}
\end{equation}
Here, $r^\prime$ is the effective rank of CFIM for single-copy POVM, and $A_\Pi$ is a semi-definite matrix depending on both the POVM $\Pi$ and the targeting state $\rho$. The exact mathematical description of $A_\Pi$ and the promotion of collective measurement, together with the proof of this theorem, are both included in Appendix\,\ref{sec:bounds}. 
Notably, in the non-degenerate case, this predicts a universal $O(\sqrt{N})$ sensitivity enhancement by simply applying collective measurement of:
\begin{equation}
\left\langle\frac{\mathrm{SNR}_{\infty}^{\mathrm{col}}(\Pi)}
     {\mathrm{SNR} {}{(\Pi)}}
\right\rangle_{A_{\Pi}}
\geq\frac{\sqrt{N}}{2}\,.
\end{equation}

 {When only a finite number \(n_s\) of copies are jointly
processed, weak commutativity alone no longer guarantees exact joint
saturation. For regular imaging working points and smoothly responding
collective receivers, the direction-averaged finite-copy result approaches its asymptotic limit
according to the following conservative array-wide scaling:
\begin{equation}
\left\langle\frac{F_{\infty}^{\mathrm{col}}(\Pi)}   {F_{n_s}^{\mathrm{col}}{(\Pi)}}
\right\rangle_{A_{\Pi}}=1\,+ O\left(\frac{N}{\sqrt{n_s}}\right)\,,
\end{equation}
where \(F_{ {},n_s}^{\mathrm{col}}\) and
\(F_{ {},\infty}^{\mathrm{col}}\) denote the directional Fisher information
of the finite-copy and asymptotic collective receivers, respectively. Meanwhile, the implicit prefactor depends on the local source state and the receiver
\(\Pi\).
Combined with the asymptotic \(O(N)\) Fisher enhancement of collective measurement compared to single-copy measurement, this result
identifies a characteristic crossover copy number \(n_s^\star\sim N^2\).
For \(n_s\lesssim N^2\), collective measurement guarantees a global Fisher gain grows as
\(O(\sqrt{n_s})\), corresponding to an SNR enhancement of
\(O(n_s^{1/4})\). For \(n_s\gtrsim N^2\), the SNR gain crosses over to the
\(O(\sqrt{N})\) regime and can nearly attain the asymptotic limit, with exact
saturation approached as \(n_s/N^2\rightarrow\infty\). Detailed derivation
of this scaling is provided in Appendix\,\ref{sec:bounds}.} 

Worth to point out, estimating either the magnitude or the real and imaginary components of \(\rho_{ij}^{(1)}\) requires only a change in the local operating point \(\theta\) of the estimators, while the operator level incompatibility stressed in Eq.\,\ref{eq:commute} always holds. Consequently, the scaling applies broadly to the entire complex-visibility estimation, rather than only restricted in phase.

\section{Physical Implementation}\label{sec:implementation}
 {We now propose a memory-assisted N-station interferometer designed to realize collective readout by coherently storing
and jointly processing source-bearing temporal modes.
The architecture builds on the  KBGL
protocol,  and, as illustrated in Fig.\,\ref{fig:weak_meas}, contains three
stages}\,\cite{Khabiboulline2019PRL}: 
\begin{enumerate}
\item \textit{Memory encoding.---}
Single photon events in a continuous stream of incoming weak thermal light are
sequentially and coherently stored into \(n_s\) parallel registers in a quantum non-demolition manner. Key technology required for such memory-assisted encoding is heralded single-photon memory, which has recently been demonstrated experimentally in kilometer-scale baselines~\cite{Stas2026,Wang2026MemoryAssistedInterferometer}. Because the memory procedures
preserve each optical state without measurement-induced collapse, the
registers naturally constitute identical copies for the subsequent
collective readout. Meanwhile, the block number \(n_s\) can be chosen flexibly to balance the desired
sensitivity enhancement against the technical complexity of the collective
readout.

\item \textit{Joint quantum processing.---}
Before detection, a coherent channel \(\mathcal C_{n_s}\) jointly mixes the \(n_s\) memory registers by a unitary or a general CPTP map. In the asymptotic regime \(n_s\to \infty\), random purification is a viable implementation route with circuit complexity polynomial in \(n_s\) and the single-copy Hilbert-space dimension~\cite{Zhou2026Purification}. For finite \(n_s\), the circuit must instead be optimized and adapted to the operating point specified by the source state \(\rho\).

\begin{figure}[t]
\centering
\makebox[\linewidth][c]{\includegraphics[width=0.99\linewidth]{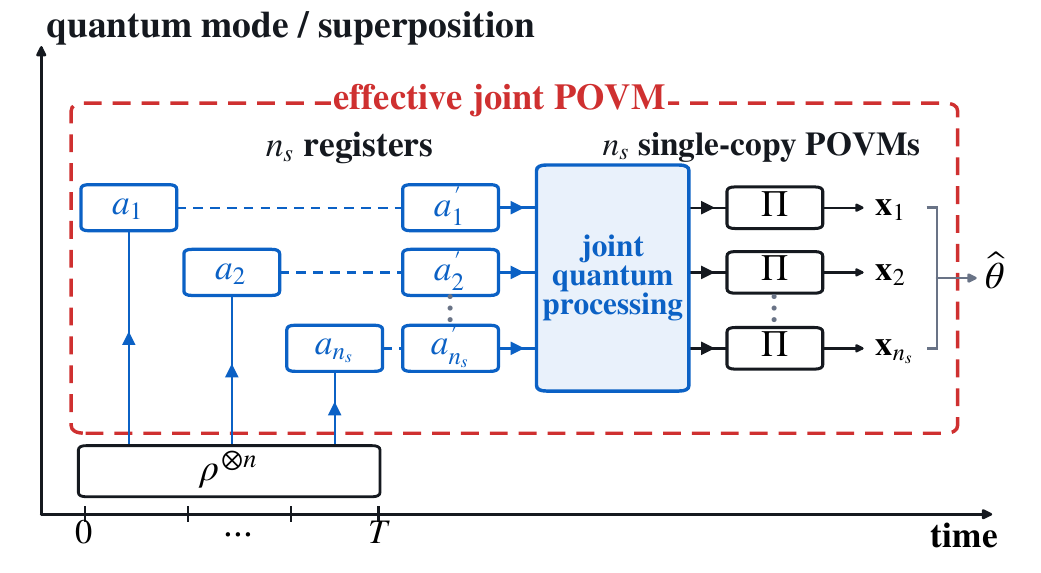}}
\caption{\label{fig:weak_meas}
 {Schematic illustration of time-unresolved measurement. Here, all the quantum operations in the red block effectively serve as a joint measurement acting on $n_s$ undistinguished copies of the quantum state.}}
\end{figure}

\item \textit{Single-copy readout.---}
The \(n_s\) output registers are measured using identical single-copy POVMs. Previous joint operation on the copies induces mutual correlations among the measurement outcomes, enabling higher estimation precision than single-copy measurements.
\end{enumerate}


In Appendix\,\ref{sec:finite}, we provide a concrete example for \(N=4\) and \(n_s=2\), in which the joint-processing layer is constructed through semi-definite optimization.

 {To connect the schematic collective-readout architecture above to a realistic multi-site interferometer, we now specify how the source coherence is mapped onto the quantum state received by the network. Over an observation time \(T\) and bandwidth \(\Delta f\), the source supplies \(n\simeq T\Delta f\) independent temporal modes sampling the same astrophysical coherence matrix. Each mode is subsequently modified by station-dependent photon collection and transport efficiencies and detector backgrounds, and may be processed either separately or through the collective readout described above. We capture these detector-level effects using the following phenomenological single-photon density matrix:}
\begin{equation}
\begin{aligned}
\rho_{ii}^{(1)}&=\frac{s_i}{s}=\frac{\eta_i u_i+\varepsilon_i}{s}\,;\\
\rho_{ij}^{(1)}&=\frac{1}{s}\sqrt{\eta_i\eta_j u_i u_j} \,g_{ij}\,,
\end{aligned}
\end{equation}
where \(s_i\) denotes the effective mean photon occupation number at station \(i\), and \(s=\sum_i s_i\) is the mean total photon occupation number across the array. Here, \(\eta_i\in[0,1]\) is the total photon efficiency from the \(i\)-th local telescope site to the final measurement stage, including photon collection efficiency, coupling/memory efficiency, and transport efficiency; \(\epsilon_i\) denotes the local background occupancy in the same detected mode, accounting for background light as well as local classical and quantum noise; and \(u_i\) is the photon occupation number in the ideal lossless case, which is determined by the array configuration and source parameters through
\begin{equation}
\begin{aligned}
u_i(f)&=\frac{\pi D_i^2}{4}\frac{3631\,{\rm Jy}}{hf}\,10^{-0.4m}\,,
\end{aligned}
\label{eq:uf_mode}
\end{equation}
with \(D_i\) being the effective aperture of i-th telescope, \(f\) being the observing frequency, and $m(f)$ being the frequency-resolved AB magnitude of the targeting source.

 {The efficiency factors \(\eta_i\) incorporate several local and nonlocal loss channels. Among them, long-distance optical transport directly determines whether the array can employ a central-hub architecture or instead requires genuinely nonlocal quantum links.} An important length scale for the interferometric array is characteristic scale of fiber attenuation, $L_0\approx 10\,\rm km$. If the signal photon is directly transported via fiber, they would be attached a transmittance loss that scales exponentially with the baseline length, \(\eta\propto e^{-L/L_0}\), while such exponential loss penalty can in principle be removed by pre-distributed entanglement or heralded quantum memories\,\cite{Gottesman2012,Khabiboulline2019PRL,Khabiboulline2019PRA}. Ground-based interferometric arrays can therefore be divided into two regimes: near-term proposals operating at \(L\lesssim L_0\), where the network could still be 
constructed via fiber-based direct signal transport and local quantum operations at the central hub, and long-term plans operating at \(L\gg L_0\), which necessitates a non-local quantum network.

Another practical issue for ground-based network is  {piston noise driven by} atmospheric turbulence, which introduces  unknown local phase $\delta_i$ to the light arriving at each individual station, modifying the off-diagonal elements by $\rho_{ij} \to \rho_{ij}e^{i(\delta_i-\delta_j)}$, thereby corrupting the intrinsic structural phase of the source. These local phase jitters $\delta_i$ are typically unknown and change rapidly at an atmospheric fluctuation timescale($\sim 10\rm \,ms$), and is hard to be fully removed by current adaptive optics techniques~\cite{Monnier2003,Monnier2007}.
Nonetheless, closure phase, defined as
\begin{equation}
\Phi_{\rm cl}
=\arg\!\left(
\rho_{ij}^{(1)}
\rho_{jk}^{(1)}
\rho_{ki}^{(1)}
\right)\,,
\label{eq:closure_definition}
\end{equation}
serves as the lowest-order gauge-invariant observable { under station local phase uncertainty and cancels station-local pistons
while retaining gauge-invariant source-structure information}\,\cite{Jennison1958,CornwellWilkinson1981,Lohmann1983,Baldwin1986,Monnier2003,Monnier2007,Thyagarajan2022,Chael2018,BroderickPesce2020}. In the following imaging simulations, we therefore retain only the closure-phase information, rather than the full set of visibility phases. The reliable amplitude and phase information of the astrophysical source therefore has dimension:
\begin{equation}
    r_{\rm amp}=\frac{N(N-1)}{2}\,;\, r_{\rm phase}=\frac{(N-1)(N-2)}{2}\,.
\end{equation}

\section{Astrophysical implication}
With these settings of a non-local collective network, we discuss the astrophysical relevance might bring by the collective measurement in both near-term and long-term schemes.

For near-term implication, we consider a benchmark imaging mission with a broadband, fiber-linked benchmark array  sited in Hawaii. The array consists of six ground-based telescopes: three existing visible-band facilities, Keck I, Subaru, and Gemini North, together with three additional 6m remote stations. The longest baseline is about \(10\,{\rm km}\), corresponding to an angular scale of  \(\sim 10\,\mu{\rm as}\) in the visible band. Light collected by each telescope is fiber-transported to a central hub. Topology of the network is shown in Fig.~\ref{fig:clean}. We assume a local photon collection efficiency of $0.02$ , fiber loss of $0.2\,\rm dB/km$, and white background occupation of $\varepsilon=10^{-9}$, 
 {which are illustrative sensitivity benchmarks feasible within current or near-future technologies.}

\begin{figure}[t]
    \centering
    \includegraphics[width=\columnwidth]{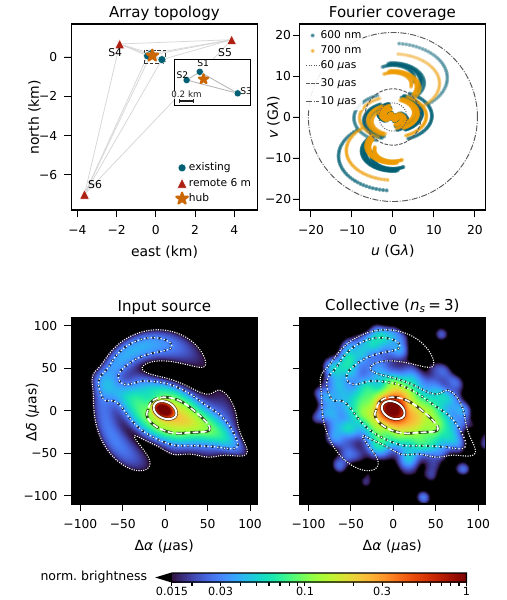}
    \caption{\label{fig:clean}
    {Closure-space imaging test in a six-station long-baseline array. \textit{Top:} Array geometry and Fourier coverage of the 6-station array. \textit{Bottom:} Benchmark targeting source and the RML-reconstructed source with 3-copy collective measurement. Dashed and dotted curves mark
equal-brightness contours of the input source overlaid on both image panels. Detailed pipeline of the RML reconstruction is supplemented in Appendix\,\ref{sec:rml}.}}
\end{figure}

\begin{figure}[t]
    \centering
    \includegraphics[width=\columnwidth]{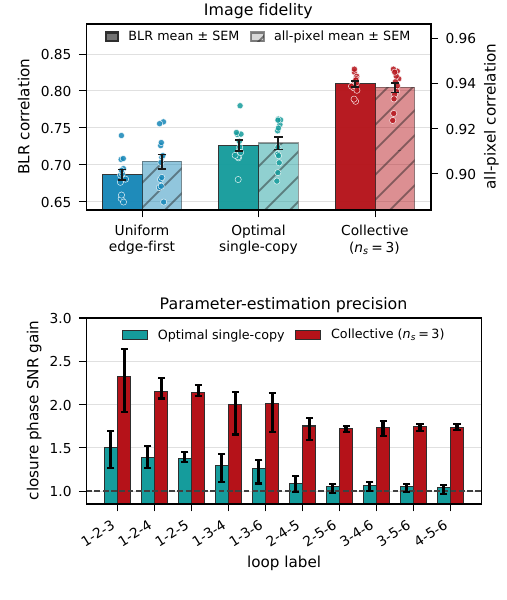}
    \caption{\label{fig:rml_stats}
    {Statistical improvement of collective-measurement-based nonlocal imaging.
\textit{Top:} Comparison of the imaging quality under different measurement schemes. Points
represents the correlation of individual imaging tests, while bars and error bars give the mean and
standard error. \textit{Bottom:} SNR gains for closure phase measurement relative to the
uniform edge-first measurement scheme; bars give geometric means and whiskers indicate
the 5--95\% sample quantiles.}}
\end{figure}

\begin{figure*}[t]
\centering
\includegraphics[width=0.94\textwidth]{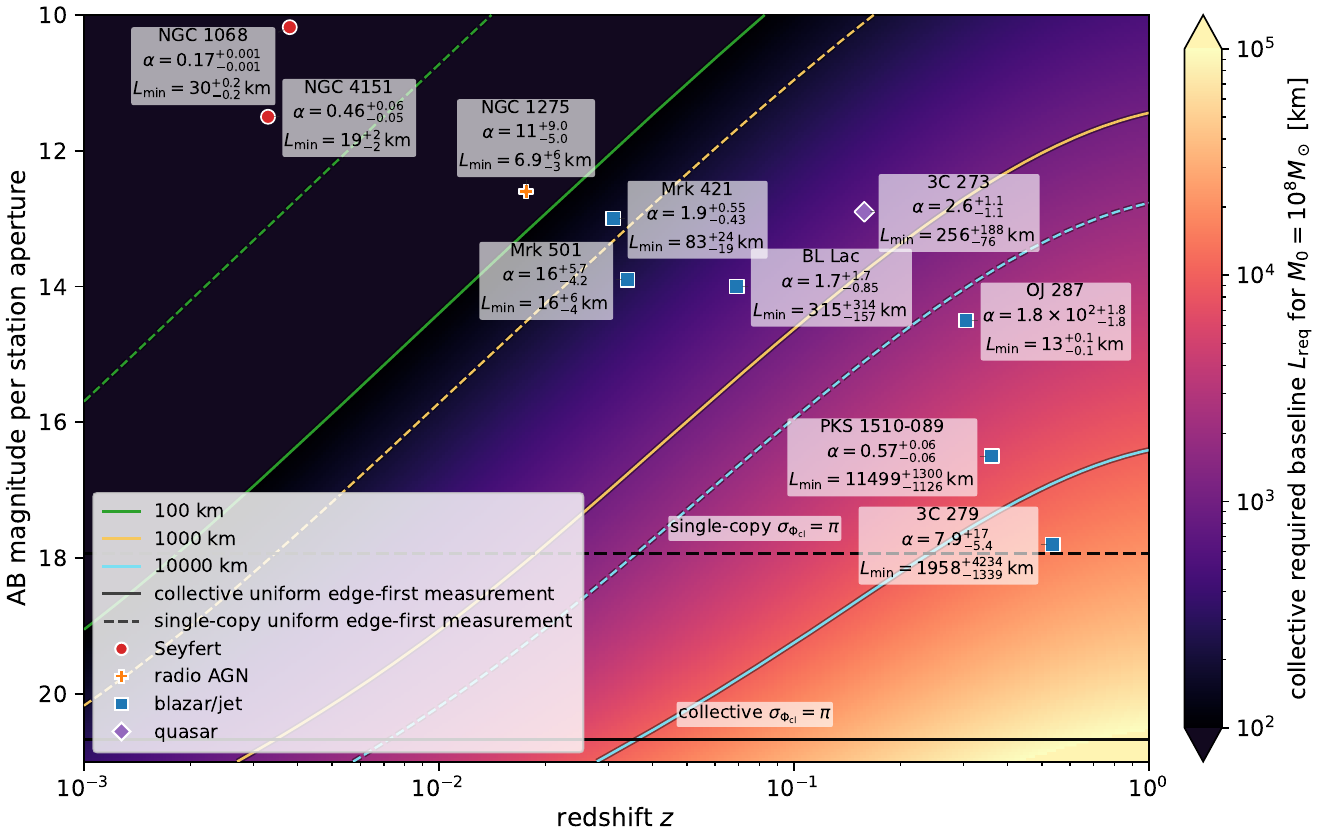} 
\caption{\label{fig:longterm_reach}
Minimum required baseline length for resolving sources at the
Schwarzschild-radius scale. The central black-hole masses adopted for the representative source markers
are taken from Refs.~\cite{Gallimore2024NGC1068,Bentz2006,
Riffel2020NGC1275,Barth2003BLLacMasses,WooUrry2002AGNMasses,
Li2022SARM3C273,Valtonen2012OJ287,Rakshit2020PKS1510,
Nilsson2009ThreeC279}. Solid and dashed black lines denote the unity closure-SNR
thresholds for single-copy measurement and collective measurement,
respectively.  For each readout strategy, the region below the corresponding
threshold is inaccessible because the closure phases cannot be estimated with
sufficient precision.}
\end{figure*}

We then simulate broadband ground-based observations across \(600\)--\(700\,{\rm nm}\) for a simulated source, with a total integration time of \(100\,{\rm ms}\) for each array position. Practically, such broadband observation could be achieved by the multi-band memory scheme proposed in Ref.\,\cite{Khabiboulline2019PRL}. We sample 36 Earth-rotation epochs separated by 15 minutes, and an RML analysis is applied to reconstruct the source image with the obtained samples of visibility amplitudes and closure phases~\cite{Chael2018}. For the benchmark source, we generate a angularly resolved source based on the Seyfert AGN NGC 4151, with average magnitude \(m_V\simeq 11.9\) in the visible band and a crescent-like broad-line-region (BLR) with angular radius of  \(\sim 70\,\mu{\rm as}\)~\cite{Bentz2006,Yuan2020}, as visualized in Fig.~\ref{fig:clean}. Reconstruction quality is scored by both the global image correlation and the correlation restricted to the BLR region  {(defined as the cylindrical region of \(\lvert r-r_{\rm BLR}\rvert<2.2\sigma_{\rm BLR}\), corresponding here to
\(45.6\,\mu{\rm as}<r<98.4\,\mu{\rm as}\))}. Further imaging details are presented in Appendix\,\ref{sec:rml}.

We perform imaging simulations for the uniform edge-first single-copy
measurement, the optimal single-copy measurement, and its promoted collective
counterpart in the three-copy setting. Their statistical reconstruction
performance is compared in Fig.~\ref{fig:rml_stats}, showing a clear improvement
of the collective measurement over both single-copy schemes. Crucially, the
representative reconstruction in Fig.~\ref{fig:clean} demonstrates that even a
three-copy collective measurement can faithfully recover detailed source
morphology, including the asymmetric disk and the extended BLR. Within this benchmark, the enhanced closure-phase precision translates directly into higher imaging quality and improves the ability to resolve characteristic structures in representative astrophysical sources, even in the presence of detector background and incomplete Fourier coverage.

For long-term astrophysical potential, we discuss how 
the sensitivity gain contributes to the fundamental goal of unveiling the evolution dynamics of astrophysical source. The finest scale for astrophysical source is defined by the Schwarzschild radius of central blackhole, $R=2GM/c^2$, corresponding to an angular scale of $\theta_{\rm tar}=R/d$ for remote sources at distance $d$. At the same time, the angular resolution of closure-phase-based interferometric array is given by: 
\begin{equation}
    \theta_{\rm res}=\frac{\sigma_{ \Phi_{\rm cl}}}{|\partial \Phi _{\rm cl}/\partial \theta|}=\frac{1}{\sqrt{nF {}}}\cdot\frac{\lambda}{2\pi L\kappa}\,,
\end{equation}
where $F$ is the Fisher information for closure phase, $\lambda$ is the operating wavelength, $L$ is the size of array, and $\kappa$ is a dimensionless structural factor: for a symmetric loop, $\kappa=0$ for a completely symmetric source distribution, and $\kappa=1$ for a fully asymmetric one like separated point-like binary. The goal-driven requirement of $\theta_{\rm res}\leq \theta_{\rm tar}$ therefore translates into requirement of baseline length via:
\begin{equation}
L\geq \frac{\lambda}{2\pi \kappa {} \sqrt{nF_ {}}}\cdot\frac{c^2d}{2GM}\,.
\end{equation}
Thus, the \(O(\sqrt{N})\) SNR gain in visibility phase readout enters the astrophysical reach in two ways. First, for a fixed baseline length, it extends the maximum distance of resolvable sources by \(O(\sqrt{N})\), corresponding to an \(O(N^{3/2})\) increase in the accessible number of sources. Second, for a specific imaging task targeting a given source at a fixed angular resolution, it reduces the required baseline length by \(O(\sqrt{N})\). 

 {To illustrate the astrophysical reach of collective readout, we first
consider the task of resolving an astrophysical favorable bright quasar 3C~273, with a dynamically inferred black-hole mass of
\(M_\bullet=(2.6\pm1.1)\times10^8M_\odot\)~\cite{Gravity2018}
and an AB magnitude of \(m_{\rm AB}=12.8\). In the numerical estimation, we assume a more futuristic setting of the array
\(N=20\), \(D_{\rm eff}=10\,{\rm m}\), \(\eta=0.2\),
\(\epsilon=10^{-11}\), and \(n=T\Delta f=10^{11}\) at
\(650\,{\rm nm}\), together with
\(\lvert g_{ij}\rvert=\kappa_{\rm cl}=0.5\) and random visibility
phases.  Under these conditions, resolving a Schwarzschild-radius-scale
structure in 3C~273 requires
\(L_{\min}^{\rm sc}\simeq1.39\times10^3\,{\rm km}\) with the
single-copy measurement, but only
\(L_{\min}^{\rm col}\simeq2.56\times10^2\,{\rm km}\) with the
collective measurement at the asymptotic limit $(n_s\to \infty)$.  Collective readout therefore reduces the
required baseline by a factor of approximately \(5.4\), moving the minimum required size of non-local array from the \(10^3\,{\rm km}\) to the \(10^2\,{\rm km}\) regime. }

 {We then extend this state-of-the-art analysis to more astrophysical sources, the result is shown in
Fig.~\ref{fig:longterm_reach}. Taking \(M_0=10^8M_\odot\) as the
reference black-hole mass, we model the minimum baseline \(L_0\)
required to resolve Schwarzschild-radius-scale structure. For an
individual source with \(M_\bullet=\alpha M_0\), the corresponding
baseline rescales as \(L=L_0/\alpha\). For \(N=20\), collective
measurement guarantees a baseline reduction of at least
\(\sqrt{(N-2)/4}\simeq2.1\), while the numerical calculation yields
reductions of approximately \(4.7\)--\(5.4\) over the displayed
magnitude range. Consequently, eight of the ten representative sources
require collective baselines below \(10^3\,{\rm km}\), while even the
fainter high-redshift targets remain within the
\(10^3\)--\(10^4\,{\rm km}\) range of terrestrial networks. The
collective advantage therefore translates directly into a broader
population of astrophysical structures that can be resolved with
Earth-scale interferometric arrays.}

\textit{Conclusion.---} In this work, we discuss the fundamental nature of imaging stars with a multi-station interferometric non-local network as a multi-parameter quantum estimation problem. We find that, all time-resolved repetitive or adaptive measurement suffer from the incompatibility of individual estimators. Such incompatibility can be nearly eliminated by incorporating collective measurement that jointly acts on multiple copies, which can bring an SNR gain of $O(\sqrt{N})$ for jointly obtaining the visibilities at the asymptotic limit, and could offer certain sensitivity enhancement even if $n_s$ is limited.
The enhanced performance for closure phases estimation could be directly converted into both the imaging quality in near-future tasks, or astrophysical potential in more futuristic regime. 

With memory-assisted nonlocal quantum links already demonstrated at the single-baseline level~\cite{Stas2026}, extending them to multi-site quantum networks is a timely and experimentally relevant next step.  {However, most concrete
protocols for multi-site nonlocal quantum interferometry
are formulated as event-by-event readout schemes, while
the station-number scaling enabled by collective temporal measurements has received less attention}~\cite{Gottesman2012,
Khabiboulline2019PRL,Khabiboulline2019PRA,Huang2022,
Marchese2023,WangChitambar2025}. Our work identifies collective measurement as a central design principle and ultimate goal for this transition, allowing favorable sensitivity scaling to be retained as the array grows.  {It thereby provides a quantum metrological framework for examining the transition from
individual links to scalable imaging arrays.}

Finally, we emphasize the implementation challenge of collective measurement. Implementing a general collective measurement on \(n_s\) copies requires much more complex quantum operation than single-copy measurement, with circuit complexity and quantum-resource overhead growing rapidly as \(n_s\) increases and making the scheme increasingly sensitive to control imperfections. However, these demands do not fully eliminate the value of our collective-measurement-based protocol.
First, both experiment and theory are rapidly expanding the range of implementable collective measurements. Experimentally, two-copy and three-copy collective measurements have been demonstrated within optical systems~\cite{Tian2024TwoCopy,Zhou2025ThreeCopy}. Theoretically, complementary proposals seek to approach the many-copy precision limit through engineered many-body interactions~\cite{Tsang2026ManyBody}. Thus, the feasibility of conducting collective measurement improves over time. Second, the sensitivity enhancement enabled by collective measurement is not all-or-nothing. Even when both \(N\) and \(n_s\) are finite and the full asymptotic \(\sqrt{N}\) enhancement is not attained, collective readout can already provide appreciable gains. Progressively incorporating collective measurements into future nonlocal interferometers and increasing the number of jointly measured copies therefore offers a possible, stepwise route toward the goal expressed in our title: imaging stars at the quantum-compatibility limit.

\section*{Acknowledgements}
We thank Yanbei Chen and Yulin Xia for the fruitful discussion about multi-parameter metrology and astrophysical relevance. We acknowledge helpful contributions from OpenAI GPT 5.6 Solar on the formalization and proof of the key metrological bounds (Eqs.7, 13 and 15). X. G. and H. M. are supported by National Natural Science Foundation of China under Grant No. 12441503 and National Key R$\&$D Program of China (2023YFC2205800). 
HY is supported by the Natural Science Foundation of China (Grant 12573048).
\vspace{0.5cm}
\section*{Data availability} 
All the data and source code that support the findings of this article are openly available in the accompanying GitHub repository\,\cite{data}.

\bibliography{references_completed}

@article{Jennison1958,
  author = {Jennison, R. C.},
  title = {A Phase Sensitive Interferometer Technique for the Measurement of the Fourier Transforms of Spatial Brightness Distributions of Small Angular Extent},
  journal = {Mon. Not. R. Astron. Soc.},
  volume = {118},
  pages = {276--284},
  year = {1958},
  doi = {10.1093/mnras/118.3.276}
}

@article{CornwellWilkinson1981,
  author = {Cornwell, T. J. and Wilkinson, P. N.},
  title = {A New Method for Making Maps with Unstable Radio Interferometers},
  journal = {Mon. Not. R. Astron. Soc.},
  volume = {196},
  pages = {1067--1086},
  year = {1981},
  doi = {10.1093/mnras/196.4.1067}
}

@article{Lohmann1983,
  author = {Lohmann, Adolf W. and Weigelt, Gerd and Wirnitzer, Bernhard},
  title = {Speckle Masking in Astronomy: Triple Correlation Theory and Applications},
  journal = {Appl. Opt.},
  volume = {22},
  pages = {4028--4037},
  year = {1983},
  doi = {10.1364/AO.22.004028}
}

@article{Baldwin1986,
  author = {Baldwin, J. E. and Haniff, C. A. and Mackay, C. D. and Warner, P. J.},
  title = {Closure Phase in High-Resolution Optical Imaging},
  journal = {Nature},
  volume = {320},
  pages = {595--597},
  year = {1986},
  doi = {10.1038/320595a0}
}

@article{BaldwinHaniff2002,
  author = {Baldwin, John E. and Haniff, Christopher A.},
  title = {The Application of Interferometry to Optical Astronomical Imaging},
  journal = {Philos. Trans. R. Soc. A},
  volume = {360},
  pages = {969--986},
  year = {2002},
  doi = {10.1098/rsta.2001.0977}
}

@article{Monnier2003,
  author = {Monnier, John D.},
  title = {Optical Interferometry in Astronomy},
  journal = {Rep. Prog. Phys.},
  volume = {66},
  pages = {789--857},
  year = {2003},
  doi = {10.1088/0034-4885/66/5/203}
}

@article{Monnier2007,
  author = {Monnier, John D.},
  title = {Phases in Interferometry},
  journal = {New Astron. Rev.},
  volume = {51},
  pages = {604--616},
  year = {2007},
  doi = {10.1016/j.newar.2007.06.006}
}

@article{Petrov2007,
  author = {Petrov, R. G. and others},
  title = {{AMBER}, the Near-Infrared Spectro-Interferometric Three-Telescope {VLTI} Instrument},
  journal = {Astron. Astrophys.},
  volume = {464},
  pages = {1--12},
  year = {2007},
  doi = {10.1051/0004-6361:20066496}
}

@article{Tatulli2007,
  author = {Tatulli, E. and others},
  title = {Interferometric Data Reduction with {AMBER/VLTI}: Principle, Estimators, and Illustration},
  journal = {Astron. Astrophys.},
  volume = {464},
  pages = {29--42},
  year = {2007},
  doi = {10.1051/0004-6361:20064799}
}

@article{Benisty2009,
  author = {Benisty, M. and Berger, J.-P. and Jocou, L. and Labeye, P. and Malbet, F. and Perraut, K. and Kern, P.},
  title = {An Integrated Optics Beam Combiner for the Second Generation {VLTI} Instruments},
  journal = {Astron. Astrophys.},
  volume = {498},
  pages = {601--613},
  year = {2009},
  doi = {10.1051/0004-6361/200811083}
}

@article{Chelli2009,
  author = {Chelli, A. and Duvert, G. and Malbet, F. and Kern, P.},
  title = {Phase Closure Nulling: Application to the Spectroscopy of Faint Companions},
  journal = {Astron. Astrophys.},
  volume = {498},
  pages = {321--327},
  year = {2009},
  doi = {10.1051/0004-6361/200811036}
}

@article{Lacour2014,
  author = {Lacour, S. and Tuthill, P. and Monnier, J. D. and Kotani, T. and Gauchet, L. and Labeye, P.},
  title = {A New Interferometer Architecture Combining Nulling with Phase Closure Measurements},
  journal = {Mon. Not. R. Astron. Soc.},
  volume = {439},
  pages = {4018--4029},
  year = {2014},
  doi = {10.1093/mnras/stu258}
}

@article{Ceus2011,
  author = {Ceus, Damien and Tonello, Alessandro and Grossard, Ludovic and Delage, Laurent and Reynaud, Fran{\c c}ois and Herrmann, Harald and Sohler, Wolfgang},
  title = {Phase Closure Retrieval in an Infrared-to-Visible Upconversion Interferometer for High Resolution Astronomical Imaging},
  journal = {Opt. Express},
  volume = {19},
  pages = {8616--8624},
  year = {2011},
  doi = {10.1364/OE.19.008616}
}

@article{Ceus2013,
  author = {Ceus, D. and Delage, L. and Grossard, L. and Reynaud, F. and Herrmann, H. and Sohler, W.},
  title = {Contrast and Phase Closure Acquisitions in Photon Counting Regime Using a Frequency Upconversion Interferometer for High Angular Resolution Imaging},
  journal = {Mon. Not. R. Astron. Soc.},
  volume = {430},
  pages = {1529--1537},
  year = {2013},
  doi = {10.1093/mnras/sts654}
}

@article{Thyagarajan2022,
  author = {Thyagarajan, Nithyanandan and Carilli, Chris L.},
  title = {A Geometric View of Closure Phases in Interferometry},
  journal = {Publ. Astron. Soc. Aust.},
  volume = {39},
  pages = {e014},
  year = {2022},
  doi = {10.1017/pasa.2022.6}
}

@article{Chael2018,
  author = {Chael, Andrew A. and Johnson, Michael D. and Bouman, Katherine L. and Blackburn, Lindy L. and Akiyama, Kazunori and Narayan, Ramesh},
  title = {Interferometric Imaging Directly with Closure Phases and Closure Amplitudes},
  journal = {Astrophys. J.},
  volume = {857},
  pages = {23},
  year = {2018},
  doi = {10.3847/1538-4357/aab6a8}
}

@article{BroderickPesce2020,
  author = {Broderick, Avery E. and Pesce, Dominic W.},
  title = {Closure Traces: Novel Calibration-Insensitive Quantities for Radio Astronomy},
  journal = {Astrophys. J.},
  volume = {904},
  pages = {126},
  year = {2020},
  doi = {10.3847/1538-4357/abbd9d}
}

@article{Bentz2006,
  author = {Bentz, Misty C. and others},
  title = {A Reverberation-Based Mass for the Central Black Hole in {NGC 4151}},
  journal = {Astrophys. J.},
  volume = {651},
  pages = {775--781},
  year = {2006},
  doi = {10.1086/507417}
}

@article{Yuan2020,
  author = {Yuan, W. and others},
  title = {The Cepheid Distance to the Seyfert 1 Galaxy {NGC 4151}},
  journal = {Astrophys. J.},
  volume = {902},
  pages = {26},
  year = {2020},
  doi = {10.3847/1538-4357/abb377}
}

@article{Gravity2018,
  author = {{GRAVITY Collaboration}},
  title = {Spatially Resolved Rotation of the Broad-Line Region of a Quasar at Sub-Parsec Scale},
  journal = {Nature},
  volume = {563},
  pages = {657--660},
  year = {2018},
  doi = {10.1038/s41586-018-0731-9}
}

@article{Tsang2011,
  author = {Tsang, Mankei},
  title = {Quantum Nonlocality in Weak-Thermal-Light Interferometry},
  journal = {Phys. Rev. Lett.},
  volume = {107},
  pages = {270402},
  year = {2011},
  doi = {10.1103/PhysRevLett.107.270402}
}

@article{Gottesman2012,
  author = {Gottesman, Daniel and Jennewein, Thomas and Croke, Sarah},
  title = {Longer-Baseline Telescopes Using Quantum Repeaters},
  journal = {Phys. Rev. Lett.},
  volume = {109},
  pages = {070503},
  year = {2012},
  doi = {10.1103/PhysRevLett.109.070503}
}

@article{Khabiboulline2019PRL,
  author = {Khabiboulline, E. T. and Borregaard, J. and De Greve, K. and Lukin, M. D.},
  title = {Optical Interferometry with Quantum Networks},
  journal = {Phys. Rev. Lett.},
  volume = {123},
  pages = {070504},
  year = {2019},
  doi = {10.1103/PhysRevLett.123.070504}
}

@article{Khabiboulline2019PRA,
  author = {Khabiboulline, E. T. and Borregaard, J. and De Greve, K. and Lukin, M. D.},
  title = {Quantum-Assisted Telescope Arrays},
  journal = {Phys. Rev. A},
  volume = {100},
  pages = {022316},
  year = {2019},
  doi = {10.1103/PhysRevA.100.022316}
}

@article{Huang2022,
  author = {Huang, Zixin and Brennen, Gavin K. and Ouyang, Yingkai},
  title = {Imaging Stars with Quantum Error Correction},
  journal = {Phys. Rev. Lett.},
  volume = {129},
  pages = {210502},
  year = {2022},
  doi = {10.1103/PhysRevLett.129.210502}
}

@article{Marchese2023,
  author = {Marchese, Marta Maria and Kok, Pieter},
  title = {Large Baseline Optical Imaging Assisted by Single Photons and Linear Quantum Optics},
  journal = {Phys. Rev. Lett.},
  volume = {130},
  pages = {160801},
  year = {2023},
  doi = {10.1103/PhysRevLett.130.160801}
}

@article{Brown2023,
  author = {Brown, Matthew R. and Allgaier, Markus and Thiel, Val{\'e}rian and Monnier, John D. and Raymer, Michael G. and Smith, Brian J.},
  title = {Interferometric Imaging Using Shared Quantum Entanglement},
  journal = {Phys. Rev. Lett.},
  volume = {131},
  pages = {210801},
  year = {2023},
  doi = {10.1103/PhysRevLett.131.210801}
}

@article{Huang2024,
  author = {Huang, Zixin and Baragiola, Ben Q. and Menicucci, Nicolas C. and Wilde, Mark M.},
  title = {Limited Quantum Advantage for Stellar Interferometry via Continuous-Variable Teleportation},
  journal = {Phys. Rev. A},
  volume = {109},
  pages = {052434},
  year = {2024},
  doi = {10.1103/PhysRevA.109.052434}
}

@article{Wang2025,
  author = {Wang, Yunkai and Zhang, Yujie and Lorenz, Virginia O.},
  title = {Astronomical Interferometry Using Continuous-Variable Quantum Teleportation},
  journal = {Phys. Rev. Res.},
  volume = {7},
  pages = {023154},
  year = {2025},
  doi = {10.1103/PhysRevResearch.7.023154}
}

@article{Stas2026,
  author = {Stas, P.-J. and others},
  title = {Entanglement-Assisted Non-Local Optical Interferometry in a Quantum Network},
  journal = {Nature},
  volume = {651},
  pages = {326--332},
  year = {2026},
  doi = {10.1038/s41586-026-10171-w}
}

@article{BraunsteinCaves1994,
  author = {Braunstein, Samuel L. and Caves, Carlton M.},
  title = {Statistical Distance and the Geometry of Quantum States},
  journal = {Phys. Rev. Lett.},
  volume = {72},
  pages = {3439--3443},
  year = {1994},
  doi = {10.1103/PhysRevLett.72.3439}
}

@article{Ragy2016,
  author = {Ragy, Sammy and Jarzyna, Marcin and Demkowicz-Dobrza{\'n}ski, Rafa{\l}},
  title = {Compatibility in Multiparameter Quantum Metrology},
  journal = {Phys. Rev. A},
  volume = {94},
  pages = {052108},
  year = {2016},
  doi = {10.1103/PhysRevA.94.052108}
}

@article{Yang2019,
  author = {Yang, Jing and Pang, Shengshi and Zhou, Yiyu and Jordan, Andrew N.},
  title = {Optimal Measurements for Quantum Multiparameter Estimation with General States},
  journal = {Phys. Rev. A},
  volume = {100},
  pages = {032104},
  year = {2019},
  doi = {10.1103/PhysRevA.100.032104}
}

@article{Albarelli2019,
  author = {Albarelli, Francesco and Friel, Jamie F. and Datta, Animesh},
  title = {Evaluating the {Holevo Cram{\'e}r--Rao} Bound for Multiparameter Quantum Metrology},
  journal = {Phys. Rev. Lett.},
  volume = {123},
  pages = {200503},
  year = {2019},
  doi = {10.1103/PhysRevLett.123.200503}
}

@article{Yamagata2013,
  author = {Yamagata, Koichi and Fujiwara, Akio and Gill, Richard D.},
  title = {Quantum Local Asymptotic Normality Based on a New Quantum Likelihood Ratio},
  journal = {Ann. Stat.},
  volume = {41},
  pages = {2197--2217},
  year = {2013},
  doi = {10.1214/13-AOS1147}
}

@article{Demkowicz2020,
  author = {Demkowicz-Dobrza{\'n}ski, Rafa{\l} and G{\'o}recki, Wojciech and Guta, Madalin},
  title = {Multi-Parameter Estimation beyond Quantum Fisher Information},
  journal = {J. Phys. A},
  volume = {53},
  pages = {363001},
  year = {2020},
  doi = {10.1088/1751-8121/ab8ef3}
}

@article{Sidhu2021,
  author = {Sidhu, Jasminder S. and Ouyang, Yingkai and Campbell, Earl T. and Kok, Pieter},
  title = {Tight Bounds on the Simultaneous Estimation of Incompatible Parameters},
  journal = {Phys. Rev. X},
  volume = {11},
  pages = {011028},
  year = {2021},
  doi = {10.1103/PhysRevX.11.011028}
}

@article{Imai2026Hierarchy,
  author = {Imai, Satoya and Yang, Jing and Pezz{\`e}, Luca},
  title = {Hierarchy of Saturation Conditions for Multiparameter Quantum Metrology Bounds},
  journal = {arXiv preprint arXiv:2602.12097},
  year = {2026},
  eprint = {2602.12097},
  archivePrefix = {arXiv},
  primaryClass = {quant-ph},
  doi = {10.48550/arXiv.2602.12097}
}

@article{Zhang2026ArraySPADE,
  author = {Zhang, Yujie and Wang, Yunkai and Wu, Wilson and Jennewein, Thomas},
  title = {Quantum-Limited Subdiffraction Telescopy Requires Genuine Multi-Telescope Interference},
  journal = {arXiv preprint arXiv:2606.27276},
  year = {2026},
  eprint = {2606.27276},
  archivePrefix = {arXiv},
  primaryClass = {quant-ph},
  doi = {10.48550/arXiv.2606.27276}
}

@article{GillMassar2000,
  author = {Gill, Richard D. and Massar, Serge},
  title = {State Estimation for Large Ensembles},
  journal = {Phys. Rev. A},
  volume = {61},
  pages = {042312},
  year = {2000},
  doi = {10.1103/PhysRevA.61.042312}
}

@article{WangChitambar2025,
  author = {Wang, Yunkai and Chitambar, Eric},
  title = {Random Distillation Protocols in Long Baseline Telescopy},
  journal = {Phys. Rev. Lett.},
  volume = {134},
  pages = {170801},
  year = {2025},
  doi = {10.1103/PhysRevLett.134.170801}
}

@article{Tian2024TwoCopy,
  author = {Tian, Boxuan and Yan, Wen-Zhe and Hou, Zhibo and Xiang, Guo-Yong and Li, Chuan-Feng and Guo, Guang-Can},
  title = {Minimum-Consumption Discrimination of Quantum States via Globally Optimal Adaptive Measurements},
  journal = {Phys. Rev. Lett.},
  volume = {132},
  pages = {110801},
  year = {2024},
  doi = {10.1103/PhysRevLett.132.110801}
}

@article{Zhou2025ThreeCopy,
  author = {Zhou, Kai and Yi, Changhao and Yan, Wen-Zhe and Hou, Zhibo and Zhu, Huangjun and Xiang, Guo-Yong and Li, Chuan-Feng and Guo, Guang-Can},
  title = {Experimental Realization of Genuine Three-Copy Collective Measurements for Optimal Information Extraction},
  journal = {Phys. Rev. Lett.},
  volume = {134},
  pages = {210201},
  year = {2025},
  doi = {10.1103/PhysRevLett.134.210201}
}

@article{Tsang2026ManyBody,
  author = {Tsang, Mankei},
  title = {Approaching the Ultimate Limit of Quantum Multiparameter Estimation by Many-Body Physics},
  journal = {arXiv preprint arXiv:2603.17955},
  year = {2026},
  eprint = {2603.17955},
  archivePrefix = {arXiv},
  primaryClass = {quant-ph},
  doi = {10.48550/arXiv.2603.17955}
}

@article{Zhou2026Purification,
  author = {Zhou, Sisi},
  title = {Quantum Metrology of Mixed States via Purification},
  journal = {arXiv preprint arXiv:2605.03975},
  year = {2026},
  eprint = {2605.03975},
  archivePrefix = {arXiv},
  primaryClass = {quant-ph},
  doi = {10.48550/arXiv.2605.03975}
}

@article{Wang2026MemoryAssistedInterferometer,
  title = {Memory-Assisted Nonlocal Interferometer toward Long-Baseline Telescopes},
  author = {Wang, Bin and Luo, Xi-Yu and Gao, Bo-Feng and Liu, Jian-Long
            and Wang, Chao-Yang and Yan, Zi and Ke, Qiao-Mu and Teng, Da
            and Zheng, Ming-Yang and Cao, Yuan and Li, Jun and Peng, Cheng-Zhi
            and Zhang, Qiang and Bao, Xiao-Hui and Pan, Jian-Wei},
  journal = {Phys. Rev. Lett.},
  volume = {136},
  pages = {240801},
  year = {2026},
  doi = {10.1103/qpzn-h7p9}
}

@article{CZ,
  title = {Temporally Localized Quantum Operations on Continuous-Wave Thermal Light},
  author = {Wang, Yunkai and Zhang, Yujie and Lorenz, Virginia O.},
  journal = {Phys. Rev. Lett.},
  volume = {135},
  issue = {11},
  pages = {113602},
  numpages = {6},
  year = {2025},
  month = {Sep},
  publisher = {American Physical Society},
  doi = {10.1103/y4v8-1wgm},
  url = {https://link.aps.org/doi/10.1103/y4v8-1wgm}
}

@article{Chen2022Hierarchy,
  author  = {Chen, Hongzhen and Chen, Yu and Yuan, Haidong},
  title   = {Information Geometry under Hierarchical Quantum Measurement},
  journal = {Physical Review Letters},
  volume  = {128},
  number  = {25},
  pages   = {250502},
  year    = {2022},
  doi     = {10.1103/PhysRevLett.128.250502},
  url     = {https://doi.org/10.1103/PhysRevLett.128.250502}
}

@article{Chen2022Incompatibility,
  author  = {Chen, Hongzhen and Chen, Yu and Yuan, Haidong},
  title   = {Incompatibility Measures in Multiparameter Quantum Estimation under Hierarchical Quantum Measurements},
  journal = {Physical Review A},
  volume  = {105},
  number  = {6},
  pages   = {062442},
  year    = {2022},
  doi     = {10.1103/PhysRevA.105.062442},
  url     = {https://doi.org/10.1103/PhysRevA.105.062442}
}

@article{SuzukiYangHayashi2020,
  author = {Suzuki, Jun and Yang, Yuxiang and Hayashi, Masahito},
  title = {Quantum State Estimation with Nuisance Parameters},
  journal = {J. Phys. A: Math. Theor.},
  volume = {53},
  pages = {453001},
  year = {2020},
  doi = {10.1088/1751-8121/ab8b78}
}

@article{DemkowiczGoreckiGuta2020,
  author = {Demkowicz-Dobrza{\'n}ski, Rafa{\l} and G{\'o}recki, Wojciech
            and Gu{\c t}{\u a}, M{\u a}d{\u a}lin},
  title = {Multi-Parameter Estimation beyond Quantum Fisher Information},
  journal = {J. Phys. A: Math. Theor.},
  volume = {53},
  pages = {363001},
  year = {2020},
  doi = {10.1088/1751-8121/ab8ef3}
}

@article{Conlon2023Collective,
  author = {Conlon, Lorc{\'a}n O. and Vogl, Tobias and Marciniak, Christian D.
            and Pogorelov, Ivan and Yung, Simon K. and Eilenberger, Falk
            and Berry, Dominic W. and Santana, Fabiana S. and Blatt, Rainer
            and Monz, Thomas and Lam, Ping Koy and Assad, Syed M.},
  title = {Approaching Optimal Entangling Collective Measurements on Quantum
           Computing Platforms},
  journal = {Nat. Phys.},
  volume = {19},
  pages = {351--357},
  year = {2023},
  doi = {10.1038/s41567-022-01875-7}
}

@book{Helstrom1976,
  author = {Helstrom, Carl W.},
  title = {Quantum Detection and Estimation Theory},
  publisher = {Academic Press},
  address = {New York},
  year = {1976},
  isbn = {978-0-12-340050-5}
}

@article{Rao1945,
  author = {Rao, C. Radhakrishna},
  title = {Information and Accuracy Attainable in the Estimation of
           Statistical Parameters},
  journal = {Bull. Calcutta Math. Soc.},
  volume = {37},
  pages = {81--91},
  year = {1945}
}

@article{Gallimore2024NGC1068,
  author = {Gallimore, Jack F. and Impellizzeri, C. M. Violette and Aghelpasand, Samaneh and Gao, Feng and Hostetter, Virginia and Lankhaar, Boy},
  title = {The Discovery of Polarized Water Vapor Megamaser Emission in a Molecular Accretion Disk},
  journal = {Astrophys. J. Lett.},
  volume = {975},
  pages = {L9},
  year = {2024},
  doi = {10.3847/2041-8213/ad864f}
}

@article{Riffel2020NGC1275,
  author = {Riffel, Rog{\'e}rio A. and others},
  title = {Ionized and Hot Molecular Outflows in the Inner 500 pc of {NGC 1275}},
  journal = {Mon. Not. R. Astron. Soc.},
  volume = {496},
  pages = {4857--4873},
  year = {2020},
  doi = {10.1093/mnras/staa1922}
}

@article{Barth2003BLLacMasses,
  author = {Barth, Aaron J. and Ho, Luis C. and Sargent, Wallace L. W.},
  title = {The Black Hole Masses and Host Galaxies of {BL Lac} Objects},
  journal = {Astrophys. J.},
  volume = {583},
  pages = {134--144},
  year = {2003},
  doi = {10.1086/345083}
}

@article{WooUrry2002AGNMasses,
  author = {Woo, Jong-Hak and Urry, C. Megan},
  title = {Active Galactic Nucleus Black Hole Masses and Bolometric Luminosities},
  journal = {Astrophys. J.},
  volume = {579},
  pages = {530--544},
  year = {2002},
  doi = {10.1086/342878}
}

@article{Li2022SARM3C273,
  author = {Li, Yan-Rong and Wang, Jian-Min and Songsheng, Yu-Yang and Zhang, Zhi-Xiang and Du, Pu and Hu, Chen and Xiao, Ming},
  title = {Spectroastrometry and Reverberation Mapping: The Mass and Geometric Distance of the Supermassive Black Hole in the Quasar {3C 273}},
  journal = {Astrophys. J.},
  volume = {927},
  pages = {58},
  year = {2022},
  doi = {10.3847/1538-4357/ac4bcb}
}

@article{Valtonen2012OJ287,
  author = {Valtonen, M. J. and Ciprini, S. and Lehto, H. J.},
  title = {On the Masses of {OJ287} Black Holes},
  journal = {Mon. Not. R. Astron. Soc.},
  volume = {427},
  pages = {77--83},
  year = {2012},
  doi = {10.1111/j.1365-2966.2012.21861.x}
}

@article{Rakshit2020PKS1510,
  author = {Rakshit, Suvendu},
  title = {Broad Line Region and Black Hole Mass of {PKS 1510-089} from Spectroscopic Reverberation Mapping},
  journal = {Astron. Astrophys.},
  volume = {642},
  pages = {A59},
  year = {2020},
  doi = {10.1051/0004-6361/202038324}
}

@article{Nilsson2009ThreeC279,
  author = {Nilsson, K. and Pursimo, T. and Villforth, C. and Lindfors, E. and Takalo, L. O.},
  title = {The Host Galaxy of {3C 279}},
  journal = {Astron. Astrophys.},
  volume = {505},
  pages = {601--604},
  year = {2009},
  doi = {10.1051/0004-6361/200912820}
}

@article{KahnGuta2009,
  author  = {Kahn, Jonas and Gu{\c{t}}{\u{a}}, M{\u{a}}d{\u{a}}lin},
  title   = {Local Asymptotic Normality for Finite Dimensional Quantum Systems},
  journal = {Communications in Mathematical Physics},
  volume  = {289},
  number  = {2},
  pages   = {597--652},
  year    = {2009},
  doi     = {10.1007/s00220-009-0787-3}
}

@article{YangChiribellaHayashi2019,
  author  = {Yang, Yuxiang and Chiribella, Giulio and Hayashi, Masahito},
  title   = {Attaining the Ultimate Precision Limit in Quantum State Estimation},
  journal = {Communications in Mathematical Physics},
  volume  = {368},
  number  = {1},
  pages   = {223--293},
  year    = {2019},
  doi     = {10.1007/s00220-019-03433-4}
}

@article{HolevoWerner2001,
  author = {Holevo, A. S. and Werner, R. F.},
  title = {Evaluating Capacities of Bosonic Gaussian Channels},
  journal = {Phys. Rev. A},
  volume = {63},
  pages = {032312},
  year = {2001},
  doi = {10.1103/PhysRevA.63.032312}
}

@article{Weedbrook2012,
  author = {Weedbrook, Christian and Pirandola, Stefano and
            Garc{\'i}a-Patr{\'o}n, Ra{\'u}l and Cerf, Nicolas J. and
            Ralph, Timothy C. and Shapiro, Jeffrey H. and Lloyd, Seth},
  title = {Gaussian Quantum Information},
  journal = {Rev. Mod. Phys.},
  volume = {84},
  pages = {621--669},
  year = {2012},
  doi = {10.1103/RevModPhys.84.621}
}

@article{Lacour2019FringeTracker,
  author = {Lacour, S. and Dembet, R. and Abuter, R. and
            F{\'e}dou, P. and Perrin, G. and Choquet, {\'E}. and
            Pfuhl, O. and Eisenhauer, F. and Woillez, J. and
            Cassaing, F. and others},
  title = {The {GRAVITY} Fringe Tracker},
  journal = {Astron. Astrophys.},
  volume = {624},
  pages = {A99},
  year = {2019},
  doi = {10.1051/0004-6361/201834981}
}

@article{Landman2024MagAOX,
  author = {Landman, R. and Haffert, S. Y. and Males, J. R. and
            Close, L. M. and Foster, W. B. and Van Gorkom, K. and
            Guyon, O. and Hedglen, A. and Kautz, M. and
            Kueny, J. K. and others},
  title = {Making the Unmodulated Pyramid Wavefront Sensor Smart:
           Closed-Loop Demonstration of Neural-Network Wavefront
           Reconstruction with {MagAO-X}},
  journal = {Astron. Astrophys.},
  volume = {684},
  pages = {A114},
  year = {2024},
  doi = {10.1051/0004-6361/202348898}
}

@article{Kornilov2012,
  author = {Kornilov, V.},
  title = {Stellar Scintillation on Large and Extremely Large Telescopes},
  journal = {Mon. Not. R. Astron. Soc.},
  volume = {426},
  pages = {647--655},
  year = {2012},
  doi = {10.1111/j.1365-2966.2012.21653.x}
}

@misc{data,
  author       = {Guo, Xinyao and Miao, Haixing and Cai, Zheng and Yang, Huan},
  title        = {Data and Source Code for ``Imaging Stars at the Quantum Compatibility Limit''},
  year         = {2026},
  publisher    = {GitHub},
  howpublished = {\url{https://github.com/xinyaoguo2024/imaging-stars-at-the-quantum-compatibility-limit}},
  url          = {https://github.com/xinyaoguo2024/imaging-stars-at-the-quantum-compatibility-limit}
}
\onecolumngrid
\clearpage
\twocolumngrid

\appendix

\makeatletter
\@addtoreset{equation}{section}
\@addtoreset{figure}{section}
\@addtoreset{table}{section}
\makeatother

\renewcommand{\theequation}{\Alph{section}\arabic{equation}}
\renewcommand{\thefigure}{\Alph{section}\arabic{figure}}
\renewcommand{\thetable}{\Alph{section}\arabic{table}}

\renewcommand{\theHsection}{appendix.\Alph{section}}
\renewcommand{\theHequation}{appendix.\Alph{section}.\arabic{equation}}
\renewcommand{\theHfigure}{appendix.\Alph{section}.\arabic{figure}}
\renewcommand{\theHtable}{appendix.\Alph{section}.\arabic{table}}

\newcommand{\rank}{\operatorname{rank}}
\newcommand{\ran}{\operatorname{ran}}
\newcommand{\I}{\mathbb I}
\newcommand{\col}{\mathrm{col}}
\newcommand{\rep}{\mathrm{rep}}
\newcommand{\dd}{\mathrm{d}}
\newcommand{\op}{\mathrm{op}}

\newtheorem{theorem}{Theorem}

\newtheorem{remark}[theorem]{Remark}

\setcounter{secnumdepth}{2}

The appendices are organized as follows:  Appendix~\ref{sec:bounds}
proves the two information bounds used in the main text.  Appendix~\ref{sec:finite}
formulates finite-copy receiver design and gives an explicit \(N=4\) example. Appendix~\ref{sec:noise_model} elaborates on the physical motivation for the loss model
and the use of closure-phase analysis in the main text.
Appendix~\ref{sec:rml} records the regularized maximum-likelihood (RML) imaging
pipeline and its twelve-seed statistical test.

\section{Derivation of Information bounds for repetitive and collective measurement}
\label{sec:bounds}

\subsection{Supplemental definition}

As a preliminary, we firstly provide a more exact definition of the POVM we consider here. For an \(N\)-station array, it contains \(E=N(N-1)/2\) distinct baselines, whose visibility phases define an \(E\)-dimensional edge-phase vector \(\boldsymbol{\phi}\). A general multi-parameter measurement is described by a POVM \(\boldsymbol{\Pi}=\{\Pi_x\}\), which maps the signal-containing density matrix \(\rho_{\boldsymbol{\phi}}^{(1)}\) onto an outcome distribution:
\begin{equation}
\Pi_x \succeq 0,\qquad
\sum_x \Pi_x=\mathbb{I}_N,\qquad
p_x(\boldsymbol{\phi})
=\operatorname{Tr}\!\left[
\rho_{\boldsymbol{\phi}}^{(1)}\Pi_x
\right]\,.
\end{equation}
And, the joint parameter estimation capability attach to the POVM is described by the CFIM $J(\Pi)$, defined as:
\begin{equation}
[J(\boldsymbol{\Pi})]_{ef}=\sum_x p_x \cdot \partial_{\phi_e} lnp_x\cdot \partial_{\phi_f} lnp_x\,.
\end{equation}

We now specify the receiver-dependent quantities used in the second bound. Fix a working point \(\boldsymbol\theta_0\) and a single-copy POVM \(\Pi=\{\Pi_x\}\), with outcome probability \(p_x(\boldsymbol\theta)=\Tr(\rho_{\boldsymbol\theta}\Pi_x)\). On its \(r'\)-dimensional estimable support, define
\begin{equation}
 \begin{aligned}
 s_a(x)&=\left.\partial_a\ln p_x(\boldsymbol\theta)\right|_{\boldsymbol\theta_0},\\
 J_{ab}(\Pi)&=\sum_xp_xs_a(x)s_b(x),\\
 X_a^\Pi&=\sum_{x,b}[J(\Pi)^{-1}]_{ab}s_b(x)\Pi_x .
 \end{aligned}
 \label{eq:end_score}
\end{equation}
The classical score \(s_a(x)\) is the local response of outcome \(x\) to parameter \(\theta_a\). It acts as a signed meter response in likelihood space: \(s_a(x)>0\) means that outcome \(x\) becomes more likely when \(\theta_a\) increases, while \(s_a(x)<0\) favors the opposite displacement, and \(|s_a(x)|\) quantifies the strength of this evidence. The scores average to zero over repeated trials, while their covariance gives the CFIM \(J(\Pi)\). The operator \(X_a^\Pi\) maps the corresponding efficient influence \(J^{-1}\boldsymbol s\) back to the signal Hilbert space. It is locally unbiased, \(\Tr(\rho X_a^\Pi)=0\) and \(\Tr[(\partial_b\rho)X_a^\Pi]=\delta_{ab}\). Its intrinsic complex covariance defines
\begin{equation}
 \begin{aligned}
 (Z_\Pi)_{ab}&=\Tr(\rho X_a^\Pi X_b^\Pi)
 =(A_\Pi)_{ab}+i(B_\Pi)_{ab},\\
 (B_\Pi)_{ab}&=\frac{1}{2i}\Tr\rho[X_a^\Pi,X_b^\Pi].
 \end{aligned}
 \label{eq:end_APi}
\end{equation}
Thus \(A_\Pi=\Re Z_\Pi\succeq0\) is the intrinsic covariance of the operator-valued score associated with \(\Pi\), before that information is irreversibly converted into the classical outcome label \(x\). Its diagonal entries quantify the irreducible fluctuation of each induced estimator, while its off-diagonal entries describe correlated fluctuations between parameter directions. By contrast, the classical covariance \(J(\Pi)^{-1}\) also contains the noise introduced by the final readout, with \(J(\Pi)^{-1}-A_\Pi\succeq0\); coherent promotion acts on the intrinsic part encoded by \(A_\Pi\). The imaginary part \(B_\Pi\) records the residual incompatibility of the same induced estimators. In the second boxed bound, \(\langle\cdot\rangle_{A_\Pi}\) denotes a uniform angular average after whitening by this metric, with \(\boldsymbol u'_\ell=A_\Pi^{1/2}\boldsymbol c_\ell/\sqrt{\boldsymbol c_\ell^\top A_\Pi\boldsymbol c_\ell}\).

For comparison, the Holevo bound for a positive cost matrix \(W\) is~\cite{Albarelli2019,Demkowicz2020}
\begin{equation}
 \begin{aligned}
 C_{\rm H}(W)=\min_{\{Y_a\}}\big\{
 &\Tr[W\Re Z(Y)]\\
 &+\|W^{1/2}\Im Z(Y)W^{1/2}\|_1\big\},\\[-2pt]
 &Z_{ab}(Y)=\Tr(\rho Y_aY_b),
 \end{aligned}
 \label{eq:end_Holevo}
\end{equation}
where the Hermitian \(Y_a\) obey \(\Tr(\rho Y_a)=0\) and \(\Tr[(\partial_b\rho)Y_a]=\delta_{ab}\). The \(A_\Pi\)-isotropic geometry naturally induces the weight \(W_\Pi=A_\Pi^{-1}\), with all inverses understood on the estimable support. The promotion jointly reads the additive fluctuations \(\mathbb F_{n_s}(X_a^\Pi)=n_s^{-1/2}\sum_{m=1}^{n_s}(X_a^\Pi)^{(m)}\), whose local Gaussian limit has covariance \(Z_\Pi\). Let \(K_\Pi=A_\Pi^{-1/2}(iB_\Pi)A_\Pi^{-1/2}\). The optimal covariant readout of this induced Gaussian score experiment has covariance
\begin{equation}
 \begin{aligned}
 V_\Pi^{\rm col}&=A_\Pi+A_\Pi^{1/2}|K_\Pi|A_\Pi^{1/2},\\
 \Tr(W_\Pi V_\Pi^{\rm col})&=r'+\Tr|K_\Pi|
 =C_{\rm H}^{(\Pi)}(W_\Pi).
 \end{aligned}
 \label{eq:end_saturation}
\end{equation}
The promoted collective measurement considered here is the QLAN pullback of this covariant readout to \(\rho^{\otimes n_s}\)~\cite{Yamagata2013}; hence it saturates the \(A_\Pi\)-induced Holevo cost in the asymptotic limit. It is therefore the Holevo-optimal joint realization naturally paired with the original single-copy receiver \(\Pi\); optimizing \(\Pi\) itself is a separate outer optimization over receiver designs.

\subsection{Proof of Eq.(7): independent-copy information capacity}

Let \(J^Q\) be the QFIM of rank \(r\), and let \(J(\Pi)\) be the CFIM of
an arbitrary one-copy POVM.  The first result is
\begin{equation}
 \left\langle\frac{F_\ell(\Pi)}{F_\ell^Q}\right\rangle_Q
 \leq \min\!\left\{1,\frac{N-1}{r}\right\}.
 \label{eq:gm_average_bound}
\end{equation}
Here the average is uniform over
\begin{equation}
 \bm u_\ell=
 \frac{(J^Q)^{-1/2}\bm c_\ell}
 {\sqrt{\bm c_\ell^{\mathsf T}(J^Q)^{-1}\bm c_\ell}}
 \in S^{r-1},
 \label{eq:qfi_sphere}
\end{equation}
and all inverses are restricted to \(\ran J^Q\).

To prove Eq.~\eqref{eq:gm_average_bound}, introduce the QFI-whitened CFIM
\(M=(J^Q)^{-1/2}J(\Pi)(J^Q)^{-1/2}\).  The Braunstein--Caves and Gill--Massar
inequalities give, respectively~\cite{BraunsteinCaves1994,GillMassar2000},
\begin{equation}
 0\preceq M\preceq I_r,
 \qquad
 \Tr M=\Tr[(J^Q)^+J(\Pi)]\leq N-1.
 \label{eq:gm_trace}
\end{equation}
For nonsingular \(M\), matrix Cauchy--Schwarz implies
\begin{equation}
 \frac{F_\ell(\Pi)}{F_\ell^Q}
 =\frac{1}{\bm u_\ell^{\mathsf T}M^{-1}\bm u_\ell}
 \leq \bm u_\ell^{\mathsf T}M\bm u_\ell\,.
\end{equation}
Since
\(\langle\bm u_\ell\bm u_\ell^{\mathsf T}\rangle=I_r/r\), spherical
averaging and Eq.~\eqref{eq:gm_trace} yield the second term in
Eq.~\eqref{eq:gm_average_bound}; the pointwise quantum
Cram\'er--Rao inequality supplies the bound of unity. For singular
\(M\), \(M^{-1}\) is understood as the Moore--Penrose pseudoinverse
\(M^{+}\) on \(\operatorname{ran}M\), while \(F_\ell(\Pi)\) is set to
zero whenever \(\bm u_\ell\notin\operatorname{ran}M\). Restricting the
Cauchy--Schwarz argument to \(\operatorname{ran}M\) then yields the same
inequality. The result also holds per copy for classically adaptive
independent-copy protocols: each conditional CFIM obeys
Eq.~\eqref{eq:gm_trace}, and the conditional information matrices add
by the Fisher-information chain rule.

\subsection{Proof of Eq.(13) : receiver-paired collective gain}

Fix a regular one-copy POVM \(\Pi=\{\Pi_x\}\), let
\(p_x=\Tr(\rho\Pi_x)>0\), and define its classical score and CFIM by
\begin{equation}
 s_a(x)=\partial_a\ln p_x,
 \qquad
 J_{ab}(\Pi)=\sum_xp_xs_a(x)s_b(x).
 \label{eq:score_cfim}
\end{equation}
On the \(r'=\rank J(\Pi)\)-dimensional support, the efficient score operators
and their intrinsic covariance are
\begin{align}
 X_a^\Pi&=\sum_{x,b}[J(\Pi)^{-1}]_{ab}s_b(x)\Pi_x,\notag\\
 (Z_\Pi)_{ab}&=\Tr(\rho X_a^\Pi X_b^\Pi)
 =(A_\Pi+iB_\Pi)_{ab}.
 \label{eq:score_operators}
\end{align}
The classical score records how sensitively each outcome probability changes,
whereas \(X_a^\Pi\) is the corresponding operator-valued, locally unbiased
influence.  Thus \(A_\Pi\) is the irreducible covariance already carried by
the receiver's compressed quantum fluctuations, while \(B_\Pi\) records their
residual noncommutativity.

The first ingredient is the score-compression capacity
\begin{equation}
 \Tr[A_\Pi J(\Pi)]\leq N-1.
 \label{eq:score_capacity}
\end{equation}
For completeness, define the unital positive map
\(Tf=\sum_xf(x)\Pi_x\) and whiten the score as
\(\bm y(x)=J(\Pi)^{-1/2}\bm s(x)\).  Then
\(\{1,y_1,\ldots,y_{r'}\}\) is orthonormal in \(L^2(p)\).  Extend it to an
orthonormal basis \(\{f_k\}\).  Parseval's identity and
\(\Pi_x^2\preceq(\Tr\Pi_x)\Pi_x\) give
\begin{align}
 \Tr[A_\Pi J(\Pi)]
 &=\sum_{a=1}^{r'}\|Ty_a\|_\rho^2 \notag\\
 &\leq\sum_x\frac{\Tr(\rho\Pi_x^2)}{p_x}-1
 \notag\\
 &\leq\sum_x\Tr\Pi_x-1=N-1,
 \label{eq:parseval_proof}
\end{align}
where \(\|Y\|_\rho^2=\operatorname{Re}\Tr(\rho Y^2)\).  Kadison's
inequality also gives \(A_\Pi\preceq J(\Pi)^{-1}\).  Hence, for
\begin{equation}
 M_\Pi=A_\Pi^{1/2}J(\Pi)A_\Pi^{1/2},
 \qquad k=\min\{r',N-1\},
 \label{eq:Mpi}
\end{equation}
one has \(0\prec M_\Pi\preceq I_{r'}\) and \(\Tr M_\Pi\leq k\).

For \(n_s\) copies, consider the additive score fluctuations
\begin{equation}
 \mathbb F_{n_s}(X_a^\Pi)
 =\frac{1}{\sqrt{n_s}}\sum_{m=1}^{n_s}(X_a^\Pi)^{(m)}.
 \label{eq:additive_score}
\end{equation}
Under quantum local asymptotic normality (QLAN), these observables converge
to a Gaussian shift model with intrinsic covariance \(Z_\Pi\)
\cite{Yamagata2013,Albarelli2019,Demkowicz2020}.  With
\(K_\Pi=A_\Pi^{-1/2}(iB_\Pi)A_\Pi^{-1/2}\), positivity of
\(A_\Pi\pm iB_\Pi\) implies \(|K_\Pi|\preceq I_{r'}\).  The optimal
covariant readout of this \(\Pi\)-induced Gaussian score experiment has
covariance
\begin{equation}
 V_\Pi=A_\Pi^{1/2}(I_{r'}+|K_\Pi|)A_\Pi^{1/2},
 \qquad A_\Pi\preceq V_\Pi\preceq2A_\Pi.
 \label{eq:gaussian_covariance}
\end{equation}
For the natural Holevo weight \(W_\Pi=A_\Pi^{-1}\), its cost is
\(\Tr(W_\Pi V_\Pi)=r'+\Tr|K_\Pi|\).  Provided the QLAN pullback is
differentiable in quadratic mean, it produces collective POVMs with per-copy
CFIM
\begin{equation}
 J^{\rm col}(\Pi)\equiv
 \lim_{n_s\rightarrow\infty}
 \frac{J_{\rm col}^{(n_s)}(\Pi)}{n_s}=V_\Pi^{-1}.
 \label{eq:per_copy_col}
\end{equation}
This is Holevo optimal for the Gaussian score experiment induced by \(\Pi\);
it is not, without a separate outer optimization, a claim of global Holevo
optimality for the complete source-state model.

Finally set
\begin{equation}
 \bm u'_\ell=
 \frac{A_\Pi^{1/2}\bm c_\ell}
 {\sqrt{\bm c_\ell^{\mathsf T}A_\Pi\bm c_\ell}},
 \qquad R_\Pi=I_{r'}+|K_\Pi|.
\end{equation}
The receiver-paired directional gain is exactly
\begin{equation}
\begin{aligned}
 \frac{F_\ell^{\rm col}(\Pi)}{F_\ell(\Pi)}
 &\equiv\frac{F_\ell[J^{\rm col}(\Pi)]}{F_\ell[J(\Pi)]}\\
 &=\frac{\bm u_\ell'^{\mathsf T}M_\Pi^{-1}\bm u'_\ell}
 {\bm u_\ell'^{\mathsf T}R_\Pi\bm u'_\ell}\\
 &\geq
 \frac{1}{(\bm u_\ell'^{\mathsf T}M_\Pi\bm u'_\ell)
 (\bm u_\ell'^{\mathsf T}R_\Pi\bm u'_\ell)}.
\end{aligned}
 \label{eq:directional_gain}
\end{equation}
The functions \(1/(xy)\) and \(1/\sqrt{xy}\) are jointly convex for
\(x,y>0\).  Jensen's inequality, \(\Tr M_\Pi\leq k\), and
\(\Tr R_\Pi\leq2r'\) therefore give
\begin{align}
 \left\langle
 \frac{F_\ell^{\rm col}(\Pi)}{F_\ell(\Pi)}
 \right\rangle_{A_\Pi}
 &\geq\frac{r'^2}{\Tr M_\Pi\Tr R_\Pi}
 \geq\frac{r'}{2k}\geq\frac{r'}{2(N-1)},\notag\\
 \left\langle
 \frac{\operatorname{SNR}_\ell^{\rm col}(\Pi)}
 {\operatorname{SNR}_\ell(\Pi)}
 \right\rangle_{A_\Pi}
 &\geq\frac{r'}{\sqrt{\Tr M_\Pi\Tr R_\Pi}}
 \geq\sqrt{\frac{r'}{2(N-1)}}.
 \label{eq:paired_average_bounds}
\end{align}
For a full-edge receiver, \(r'=N(N-1)/2\), yielding the universal lower
bounds \(N/4\) in Fisher information and \(\sqrt N/2\) in SNR.  These are
local, asymptotic QLAN statements; finite-\(n_s\) performance must be computed
from an explicit joint POVM.

\subsection{  {state-of-the-art performance of finite-copy collective measurement}}
We use the score operators \(X_a^\Pi\), their intrinsic covariance \(A_\Pi\), and the
asymptotic promoted covariance \(V_{\Pi,\infty}\) defined in the End Matter. The index
\(a=1,\ldots,r\) runs over the identifiable real degrees of freedom of all off-diagonal
entries of \(\rho^{(1)}\). No closure-space projection is made here; hence
\begin{equation}
 r\leq 2E=N(N-1),
 \label{eq:sm_rank}
\end{equation}
with \(r\leq E=N(N-1)/2\) if only one quadrature of each edge coherence is estimated.
We first assume that the observation supplies enough photons to repeat an \(n_s\)-copy
receiver many times. Finite photon number is restored only at the end.

The asymptotic construction replaces the additive score fluctuations by a quantum Gaussian
shift model. At finite \(n_s\), however, the commutator of two collective scores is still an
operator-valued sample average rather than its ensemble mean. We show first that this
residual noncommutativity decreases as \(n_s^{-1/2}\), and then transfer that estimate to the
same \(A_\Pi\)-weighted directional Fisher ratio used in the main text.

Whiten the single-copy influence operators in their intrinsic metric,
\begin{equation}
 Y_a=\sum_b(A_\Pi^{-1/2})_{ab}X_b^\Pi,
 \label{eq:sm_whitened_score}
\end{equation}
and define their normalized block fluctuations
\begin{equation}
 \mathbb F_{n_s}(Y_a)=\frac{1}{\sqrt{n_s}}
 \sum_{m=1}^{n_s}Y_a^{(m)}.
 \label{eq:sm_fluctuation}
\end{equation}
Here \(Y_a^{(m)}\) acts on the \(m\)th copy. Operators belonging to different copies
commute, so
\begin{align}
 [\mathbb F_{n_s}(Y_a),\mathbb F_{n_s}(Y_b)]
 &=\frac{1}{n_s}\sum_{m,m'=1}^{n_s}[Y_a^{(m)},Y_b^{(m')}]
 \notag\\
 &=\frac{1}{n_s}\sum_{m=1}^{n_s}C_{ab}^{(m)},
 \qquad C_{ab}\equiv[Y_a,Y_b].
 \label{eq:sm_commutator_sum}
\end{align}
The Gaussian limit keeps only
\(\bar C_{ab}=\Tr(\rho C_{ab})I\). The finite-block remainder is therefore
\begin{equation}
 \Delta C_{ab}^{(n_s)}=
 \frac1{n_s}\sum_{m=1}^{n_s}
 \left(C_{ab}^{(m)}-\bar C_{ab}\right).
 \label{eq:sm_delta_comm}
\end{equation}

\begin{lemma}[Exact self-averaging law]
With \(\|Z\|_{2,\rho}^2=\Tr(\rho Z^\dagger Z)\),
\begin{equation}
 \|\Delta C_{ab}^{(n_s)}\|_{2,\rho^{\otimes n_s}}^2
 =\frac1{n_s}\|C_{ab}-\bar C_{ab}\|_{2,\rho}^2.
 \label{eq:sm_exact_pair}
\end{equation}
\end{lemma}

\begin{proof}
Writing \(\delta C_{ab}=C_{ab}-\bar C_{ab}\) and expanding the norm gives
\begin{align}
 \|\Delta C_{ab}^{(n_s)}\|_{2,\rho^{\otimes n_s}}^2
 &=\frac1{n_s^2}\sum_{m,m'=1}^{n_s}
 \Tr\!\left[\rho^{\otimes n_s}
 (\delta C_{ab}^{(m)})^\dagger\delta C_{ab}^{(m')}\right].
 \label{eq:sm_norm_expansion}
\end{align}
For \(m\neq m'\), the product state factorizes the trace into
\(\Tr(\rho\delta C_{ab}^\dagger)\Tr(\rho\delta C_{ab})=0\). The \(n_s\)
diagonal terms are all equal to \(\Tr(\rho\delta C_{ab}^\dagger\delta C_{ab})\).
Thus Eq.~\eqref{eq:sm_norm_expansion} contains \(n_s\) surviving terms divided by
\(n_s^2\), which proves Eq.~\eqref{eq:sm_exact_pair}.
\end{proof}

To collect all estimated directions without introducing a worst-case edge pair, define
\begin{equation}
 \kappa_\Pi^2\equiv
 \frac{2}{r(r-1)}\sum_{a<b}
 \|C_{ab}-\bar C_{ab}\|_{2,\rho}^2.
 \label{eq:sm_kappa_definition}
\end{equation}
This is not an additional matrix inequality: it is the root-mean-square strength of the
centered commutator for a uniformly selected pair of \(A_\Pi\)-normalized score directions.
Its physical meaning is direct. If \(\kappa_\Pi=0\), the relevant score operators commute
at the operator level and finite blocks carry no residual incompatibility. If
\(\kappa_\Pi=O(1)\), a typical noncommuting pair has a finite single-copy fluctuation, but
that fluctuation does not grow with the array size after whitening. A large
\(\kappa_\Pi\) signals either unusually strong score noncommutativity or an ill-conditioned
working point at which \(A_\Pi^{-1/2}\) greatly amplifies one direction.

The aggregate residual per estimated coordinate is
\begin{align}
 \epsilon_{n_s,\Pi}^2
 &\equiv\frac1r\sum_{a<b}
 \|\Delta C_{ab}^{(n_s)}\|_{2,\rho^{\otimes n_s}}^2
 \notag\\
 &=\frac1{rn_s}\sum_{a<b}
 \|C_{ab}-\bar C_{ab}\|_{2,\rho}^2
 =\frac{r-1}{2n_s}\kappa_\Pi^2,
 \label{eq:sm_epsilon_calculation}
\end{align}
or
\begin{equation}
 \epsilon_{n_s,\Pi}=\kappa_\Pi
 \sqrt{\frac{r-1}{2n_s}}.
 \label{eq:sm_epsilon_result}
\end{equation}
For the full edge family, \(\sqrt r=O(N)\), giving the conservative scale
\(O(N/\sqrt{n_s})\). The estimate can be tighter when the score algebra is sparse. For
example, if each normalized edge score fails to commute with only \(\Delta\) other scores,
then the same calculation contains only \(O(r\Delta)\) nonzero pairs and gives
\(\epsilon_{n_s,\Pi}=O(\sqrt{\Delta/n_s})\). The displayed \(N/\sqrt{n_s}\) law is
therefore an array-wide upper scaling, not a claim that every source and receiver saturates it.

The self-averaging law isolates the array-size dependence in \(r\) and the local
state dependence in \(\kappa_\Pi\). To make the latter distinction concrete, we now relate
\(\kappa_\Pi\) to the astronomical visibility. In general, \(\kappa_\Pi\) depends on the
complete visibility matrix, not only on a single contrast parameter.
Equation~\eqref{eq:sm_kappa_definition} can be written as
\begin{equation}
 \kappa_\Pi^2(\rho)=\frac{2}{r(r-1)}\sum_{a<b}
 \left\{\Tr(\rho C_{ab}^\dagger C_{ab})
 -|\Tr(\rho C_{ab})|^2\right\},
 \label{eq:sm_kappa_state}
\end{equation}
where both \(C_{ab}=[Y_a,Y_b]\) and the expectation value depend on
\(g_{ij}\). The first dependence enters through the whitening
\(Y=A_\Pi^{-1/2}X^\Pi\), because \(A_\Pi\) changes with the source; the second enters
directly through the state average in Eq.~\eqref{eq:sm_kappa_state}. Thus there is no
receiver-independent function \(\kappa(g)\) for an arbitrary source and POVM.

The dependence becomes explicit for the symmetric imaging benchmark used throughout the
scaling discussion. Take equal station populations and a common real visibility
\(0<g<1\),
\begin{equation}
 \rho_g=\frac{1-g}{N}I_N+g|b\rangle\langle b|,
 \qquad |b\rangle=\frac1{\sqrt N}\sum_i|i\rangle,
 \label{eq:sm_symmetric_state}
\end{equation}
and use the uniform edge-first phase receiver at its quadrature working point. Orient each
edge \(e=(ij)\), put
\(Q_e=i(|i\rangle\langle j|-|j\rangle\langle i|)\), and let \(B\) be the oriented
vertex--edge incidence matrix. A direct evaluation of the single-edge score gives
\begin{equation}
 X_e^\Pi=\frac{N}{2g}Q_e,
 \qquad
 A_\Pi=\frac{N}{4g^2}
 \left[2(1-g)I_E+gB^{\mathsf T}B\right].
 \label{eq:sm_symmetric_A}
\end{equation}
For completeness, the diagonal term follows from
\(\Tr(\rho_gQ_e^2)=2/N\). For two oriented edges, the cross term is proportional to their
incidence overlap,
\(\Tr(\rho_gQ_eQ_f)=g(B^{\mathsf T}B)_{ef}/N\) for \(e\neq f\), which yields the second
relation in Eq.~\eqref{eq:sm_symmetric_A}.

The incidence matrix separates the edge space into an \((N-1)\)-dimensional station-gradient
sector and an \([E-(N-1)]\)-dimensional divergence-free sector. On the latter,
\(B\bm q=0\), so Eq.~\eqref{eq:sm_symmetric_A} has eigenvalue
\begin{equation}
 a_{\rm dark}=\frac{N(1-g)}{2g^2}.
 \label{eq:sm_dark_variance}
\end{equation}
Consequently, the whitened score in this dominant \(O(N^2)\)-dimensional sector carries the
factor
\begin{equation}
 Y_\mu=\sqrt{\frac{N}{2(1-g)}}\,T_\mu,
 \label{eq:sm_dark_whitening}
\end{equation}
where the \(T_\mu\) are unit antisymmetric generators on the \((N-1)\)-dimensional dark-mode
space. Each generator fails to commute with \(2(N-3)\) others. For every nonzero pair,
the commutator is another unit generator and
\begin{equation}
 \|[Y_\mu,Y_\nu]\|_{2,\rho_g}^2
 =\frac{N}{2(1-g)}.
 \label{eq:sm_nonzero_pair_g}
\end{equation}
The fraction of noncommuting unordered pairs is \(4/N\). Substitution into
Eq.~\eqref{eq:sm_kappa_definition} therefore gives, exactly within this dominant sector,
\begin{equation}
\kappa_\Pi^2=\frac{2}{1-g},\qquad
 \kappa_\Pi=\sqrt{\frac{2}{1-g}}.
 \label{eq:sm_kappa_g}
\end{equation}
The remaining station-gradient sector contains only \(O(N)\) of the \(O(N^2)\) edge
directions. Including it changes the spherical average only at relative order \(1/N\), so
for the complete edge space
\begin{equation}
 \kappa_{\Pi,\rm edge}^2
 =\frac{2}{1-g}\left[1+O(N^{-1})\right]
 \label{eq:sm_kappa_full_edge}
\end{equation}
for this symmetric receiver and working point.

The factor \((1-g)^{-1/2}\) has a simple physical origin. As \(g\to1\), the one-photon
state approaches a bright pure mode, while the intrinsic variance of the orthogonal edge
combinations in Eq.~\eqref{eq:sm_dark_variance} vanishes as \(1-g\). Normalizing those
directions to unit intrinsic variance therefore amplifies their residual commutators by
\((1-g)^{-1/2}\). For the value \(g=0.5\) used in the numerical benchmark,
Eq.~\eqref{eq:sm_kappa_g} gives \(\kappa_\Pi=2\), an order-unity constant. Near unity,
however, the finite-depth correction is more accurately written as
\begin{equation}
 O\!\left(\frac{N}{\sqrt{n_s(1-g)}}\right),
 \qquad n_s(1-g)\gtrsim N^2
 \label{eq:sm_g_corrected_scaling}
\end{equation}
for the onset of the asymptotic regime. At the opposite endpoint \(g\to0\), the normalized
commutator factor remains finite, but the phase Fisher information itself vanishes; the phase
coordinate is not identifiable at \(g=0\).

The calculation above identifies the operator fluctuation that remains after a finite
number of copies have been mixed. We next transfer this algebraic residual to the same
directional Fisher metric used in the main text. Let \(J_{\Pi,n_s}^{\rm col}\) be the per-copy
CFIM of the finite-depth promoted receiver and write its covariance relative to the
asymptotic promoted covariance as
\begin{equation}
 [J_{\Pi,n_s}^{\rm col}]^{-1}
 =V_{\Pi,\infty}+A_\Pi^{1/2}D_{\Pi,n_s}A_\Pi^{1/2}.
 \label{eq:sm_covariance_remainder}
\end{equation}
The dimensionless matrix \(D_{\Pi,n_s}\) is the readout error expressed in the same intrinsic
coordinates as Eq.~\eqref{eq:sm_whitened_score}. A smooth implementation of the Gaussian
covariant receiver has a first-order response
\begin{equation}
 \frac1r\Tr|D_{\Pi,n_s}|
 \leq L_\Pi\epsilon_{n_s,\Pi}
 +O(\epsilon_{n_s,\Pi}^2).
 \label{eq:sm_receiver_stability}
\end{equation}
Here \(L_\Pi\) is the susceptibility of the joint readout to a small residual in its score
algebra. Physically, Eq.~\eqref{eq:sm_receiver_stability} excludes a circuit tuned exactly to
a singular threshold, where an infinitesimal change of the collective score moments would
produce a finite change of the outcome likelihood derivatives. Away from such thresholds,
both the POVM probabilities and their derivatives vary smoothly, so \(L_\Pi\) remains finite.
Thus \(\kappa_\Pi\) describes the fluctuation supplied by the imaging model, whereas
\(L_\Pi\) describes how strongly the chosen receiver converts that fluctuation into estimation
noise. Qualitative QLAN alone guarantees neither a finite value uniform in \(N\) nor a Fisher
rate; Eq.~\eqref{eq:sm_receiver_stability} is the explicit derivative-level regularity needed
for the latter \cite{KahnGuta2009,Yamagata2013,YangChiribellaHayashi2019}.

We now derive the directional bound. Define, as in the main text,
\begin{align}
 \bm u_\ell&=\frac{A_\Pi^{1/2}\bm c_\ell}
 {\sqrt{\bm c_\ell^{\mathsf T}A_\Pi\bm c_\ell}},\notag\\
 R_\Pi&=A_\Pi^{-1/2}V_{\Pi,\infty}A_\Pi^{-1/2}
 =I+|K_\Pi|.
 \label{eq:sm_u_R}
\end{align}
Since \(I\preceq R_\Pi\preceq2I\), Eqs.~\eqref{eq:sm_covariance_remainder} and
\eqref{eq:sm_u_R} give
\begin{align}
 \bm c_\ell^{\mathsf T}[J_{\Pi,n_s}^{\rm col}]^{-1}\bm c_\ell
 &=\left(\bm c_\ell^{\mathsf T}A_\Pi\bm c_\ell\right)
 \bm u_\ell^{\mathsf T}(R_\Pi+D_{\Pi,n_s})\bm u_\ell,
 \notag\\
 \bm c_\ell^{\mathsf T}V_{\Pi,\infty}\bm c_\ell
 &=\left(\bm c_\ell^{\mathsf T}A_\Pi\bm c_\ell\right)
 \bm u_\ell^{\mathsf T}R_\Pi\bm u_\ell.
 \label{eq:sm_two_variances}
\end{align}
Using \(F_\ell[J]=(\bm c_\ell^{\mathsf T}J^{-1}\bm c_\ell)^{-1}\), their ratio is
therefore exactly
\begin{equation}
 \frac{F_{\ell,\infty}^{\rm col}(\Pi)}
 {F_{\ell,n_s}^{\rm col}(\Pi)}
 =1+\frac{\bm u_\ell^{\mathsf T}D_{\Pi,n_s}\bm u_\ell}
 {\bm u_\ell^{\mathsf T}R_\Pi\bm u_\ell}.
 \label{eq:sm_exact_ratio}
\end{equation}
Because the denominator is at least one,
\begin{equation}
 \delta_{n_s}\equiv
 \left\langle\frac{F_{\ell,\infty}^{\rm col}}
 {F_{\ell,n_s}^{\rm col}}\right\rangle_{A_\Pi}-1
 \label{eq:sm_delta_definition}
\end{equation}
obeys
\begin{align}
 |\delta_{n_s}|
 &\leq\int_{S^{r-1}}\dd\bm u\,
 |\bm u^{\mathsf T}D_{\Pi,n_s}\bm u|
 \notag\\
 &\leq\int_{S^{r-1}}\dd\bm u\,
 \bm u^{\mathsf T}|D_{\Pi,n_s}|\bm u
 =\frac1r\Tr|D_{\Pi,n_s}|.
 \label{eq:sm_spherical_calculation}
\end{align}
The last equality follows from
\(\int_{S^{r-1}}\dd\bm u\,\bm u\bm u^{\mathsf T}=I/r\). Substitution of
Eqs.~\eqref{eq:sm_epsilon_result} and \eqref{eq:sm_receiver_stability} proves
\begin{equation}
 \boxed{
 \left\langle\frac{F_{\ell,\infty}^{\rm col}(\Pi)}
 {F_{\ell,n_s}^{\rm col}(\Pi)}\right\rangle_{A_\Pi}
 =1+O\!\left(\sqrt{\frac{r}{n_s}}\right)
 =1+O\!\left(\frac{N}{\sqrt{n_s}}\right).}
 \label{eq:sm_main_result}
\end{equation}
More explicitly,
\begin{equation}
 \left|\left\langle\frac{F_{\ell,\infty}^{\rm col}}
 {F_{\ell,n_s}^{\rm col}}\right\rangle_{A_\Pi}-1\right|
 \leq L_\Pi\kappa_\Pi\sqrt{\frac{r-1}{2n_s}}
 +O\!\left(\frac{r}{n_s}\right).
 \label{eq:sm_explicit_result}
\end{equation}
The constants may depend on the local source state and on \(\Pi\). Uniform scaling with
\(N\) requires a regular sequence of operating points for which the relevant eigenvalues and
the conditioning of \(A_\Pi\) remain controlled. This is the natural imaging analogue of the
regularity conditions used in quantitative local asymptotic normality.

The resulting correction also determines when the asymptotic collective advantage becomes
accessible. For a nondegenerate edge family, the main result of the main text gives an asymptotic
collective-to-single-copy Fisher gain \(G_F(\infty)=O(N)\). Combining it with
Eq.~\eqref{eq:sm_main_result} gives the finite-depth scaling envelope
\begin{equation}
 G_F(n_s)\sim\frac{O(N)}{1+O(N/\sqrt{n_s})}.
 \label{eq:sm_gain_envelope}
\end{equation}
It identifies \(n_s^\star\sim N^2\) as the conservative characteristic depth. When
\(1\ll n_s\lesssim N^2\), the finite-depth correction can dominate, yielding
\(G_F(n_s)=O(\sqrt{n_s})\) and hence an SNR enhancement
\(G_{\rm SNR}(n_s)=O(n_s^{1/4})\). For \(n_s\gtrsim N^2\), the correction becomes order
unity or smaller and the receiver can approach the asymptotic
\(G_{\rm SNR}=O(\sqrt N)\) scaling.

Finally, let \(n_\gamma\) be the number of detected, parameter-bearing one-photon events.
Grouping them into blocks gives \(B=\lfloor n_\gamma/n_s\rfloor\) independent joint
outcomes. Because their CFIMs add,
\begin{equation}
 F_{\ell,n_s}^{\rm tot}=n_\gamma F_{\ell,n_s}^{\rm col}
 \left[1+O\!\left(\frac{n_s}{n_\gamma}\right)\right].
 \label{eq:sm_finite_photons}
\end{equation}
Finite brightness therefore changes the number of available repetitions, not the intrinsic
ratio in Eq.~\eqref{eq:sm_main_result}. The collective-depth analysis applies in the window
\(1\ll n_s\ll n_\gamma\), with \(n_s\gtrsim N^2\) additionally required only when one aims
to approach the asymptotic \(\sqrt N\) enhancement.

\section{Concrete design of a Finite-copy collective measurement}
\label{sec:finite}

\subsection{semidefinite-optimization workflow}

Let \(\rho_{\bm\phi}\) be the informative one-photon state and define the
local block coordinate \(\bm h=\sqrt{n_s}(\bm\phi-\bm\phi_0)\).  At the known
working point,
\begin{align}
 \sigma&=\rho_0^{\otimes n_s},\notag\\
 \dot\sigma_e&=\frac{1}{\sqrt{n_s}}
 \sum_{m=1}^{n_s}\rho_0^{\otimes(m-1)}\otimes
 (\partial_e\rho)_0\otimes\rho_0^{\otimes(n_s-m)}.
 \label{eq:block_tangents}
\end{align}
For a joint POVM \(\{M_y\}\), define
\begin{equation}
 p_y=\Tr(\sigma M_y),\qquad
 \bm d_y=\Tr(\dot{\bm\sigma}M_y),\qquad
 J=\sum_y\frac{\bm d_y\bm d_y^{\mathsf T}}{p_y}.
 \label{eq:block_cfim}
\end{equation}
Because of the \(\sqrt{n_s}\) rescaling, \(J\) is the CFIM per input copy.
We use the A-optimal risk \(\mathcal R_W=\Tr(WJ^{-1})\), with \(W=I\) in the
example below.  Thus a numerical design run requires only
\((N,n_s,\rho_0,\{(\partial_e\rho)_0\},W)\), together with any imposed receiver
symmetry or pairing constraint.  Its output is an explicit list of effects
\(\{M_y\}\), not only a value of the bound: Eq.~\eqref{eq:block_cfim} gives
every outcome probability, likelihood derivative, likelihood score
\(\bm\ell_y=\bm d_y/p_y\), and the achieved CFIM.

For fixed candidate likelihood-score labels \(\bm\ell_y\), the effect optimization is one
convex semidefinite program.  Introducing \(G\) and an epigraph \(V\), solve
\begin{equation}
\begin{aligned}
 \underset{\{M_y,p_y\},G,V}{\operatorname{minimize}}\quad
 &\Tr(WV)\\
 \text{subject to}\quad
 &M_y\succeq0,\quad \sum_yM_y=I,\quad p_y=\Tr(\sigma M_y),\\
 &\Tr(\dot\sigma_eM_y)=p_y \ell_{y,e},\quad
 G=\sum_yp_y\bm\ell_y\bm\ell_y^{\mathsf T},\\
 &G\succeq\gamma I,\qquad
 \begin{pmatrix}V&I\\I&G\end{pmatrix}\succeq0.
 \label{eq:fixed_score_sdp}
\end{aligned}
\end{equation}
At the optimum \(G\) is the CFIM and the Schur complement makes
\(V\succeq G^{-1}\).  If the receiver must remain strictly paired with a
given one-copy POVM, its promoted influence operators
\(F_e^{(n_s)}=n_s^{-1/2}\sum_m(X_e^\Pi)^{(m)}\) are retained by adding
\begin{equation}
 \sum_y \ell_{y,e}M_y=\sum_fG_{ef}F_f^{(n_s)}.
 \label{eq:strict_pairing}
\end{equation}

All coefficients in Eqs.~\eqref{eq:fixed_score_sdp} and
\eqref{eq:strict_pairing} are known matrices at the supplied working point.
Consequently, the inner step can be passed directly to a standard SDP solver:
positivity, completeness, Born probabilities, local response, and the A-risk
epigraph are imposed simultaneously.  The returned matrices are the joint
effects in the chosen basis.  This step is globally optimal for the supplied
score frame; it does not yet optimize over all physically allowed scores.

The complete design is a column-generation loop around
Eq.~\eqref{eq:fixed_score_sdp}, because the self-consistent scores
\(\bm\ell_y=\bm d_y/p_y\) depend on the effects.  A run uses the following
operational sequence.

\begin{enumerate}
 \item \textit{Build the local statistical model.---}Construct \(\sigma\) and
 all \(\dot\sigma_e\) from Eq.~\eqref{eq:block_tangents}.  Remove null
 directions of the tangent Gram matrix and express \(W\) on the remaining
 estimable support; this prevents an artificial singularity in \(J^{-1}\).

 \item \textit{Reduce by copy symmetry.---}Transform \(\sigma\) and
 \(\dot\sigma_e\) into copy-permutation Schur blocks.  Only the reduced
 matrices are optimization variables; multiplicity spaces are included as
 known weights.  In the two-copy \(N=4\) problem, one \(16\times16\) effect is
 thereby replaced by independent \(10\times10\) and \(6\times6\) blocks.

 \item \textit{Initialize with an exactly feasible receiver.---}For a
 receiver paired to a one-copy POVM \(\Pi=\{\Pi_x\}\), start from the
 coarse-grained product POVM \(\Pi^{\otimes n_s}\).  A block record
 \(\bm x=(x_1,\ldots,x_{n_s})\) has the known score
 \begin{equation}
  \bm\ell_{\bm x}^{(0)}=\frac{1}{\sqrt{n_s}}
  \sum_{m=1}^{n_s}\bm\ell_{x_m}.
  \label{eq:product_score}
 \end{equation}
 Outcomes related by permuting the copies are merged without changing
 completeness.  Thus 12 one-copy outcomes give 78 unordered labels at
 \(n_s=2\), rather than 144 ordered records.  For an unrestricted search the
 same product receiver is a convenient seed, but
 Eq.~\eqref{eq:strict_pairing} is omitted.

 \item \textit{Optimize the current score frame.---}Warm-start and solve
 Eq.~\eqref{eq:fixed_score_sdp}.  Recompute \(p_y,\bm d_y,J\) directly from
 the returned effects, update every active score to \(\bm d_y/p_y\), and
 prune negligible effects only after their probabilities and constraint
 residuals have been recorded.

 \item \textit{Price new collective outcomes.---}Use the least-informative
 generalized eigenvectors of \((J,W)\), the coordinate axes, and random trial
 directions \(\bm u\).  For each direction, solve
 \begin{equation}
  \Bigl(\sum_eu_e\dot\sigma_e\Bigr)|\psi\rangle
  =\lambda\sigma|\psi\rangle,\quad
  \ell_e(\psi)=\frac{\langle\psi|\dot\sigma_e|\psi\rangle}
  {\langle\psi|\sigma|\psi\rangle}.
  \label{eq:score_pricing_operational}
 \end{equation}
 The extremal generalized eigenvectors give boundary points of the physical
 score set.  Append these labels with zero initial effects and resolve the
 SDP.  The preceding solution remains feasible, so the optimized risk cannot
 increase.

 \item \textit{Stop, reconstruct, and validate.---}Iterate until the relative
 risk decrease, SDP primal--dual gap, and best tested pricing improvement are
 below the chosen tolerances.  Reconstruct the full-space effects and report
 \begin{equation}
  \epsilon_{\rm comp}=\left\|\sum_yM_y-I\right\|_{\rm F},\quad
  \epsilon_J=\left\|J-\sum_y\frac{\bm d_y\bm d_y^{\mathsf T}}{p_y}
  \right\|_{\rm F},
  \label{eq:numerical_residuals}
 \end{equation}
 together with the minimum \(p_y\), the smallest positivity eigenvalue, and
 copy-permutation commutators.  We also verify \(J\preceq J^Q\) per copy.
 These checks distinguish a realizable receiver from a favorable but
 infeasible Fisher matrix.
\end{enumerate}

In the accompanying implementation, the convex subproblem is written in
CVXPY and solved by SCS with default accuracy \(5\times10^{-6}\) and at most
\(10^5\) iterations.  The direct PVM driver exposes the random seed, number
of Haar restarts, and iteration cap; the supplied drivers default to 20
restarts and \(1800\)--\(2000\) descent steps.  A probability cutoff of
\(10^{-8}\) is used only to
flag inactive outcomes.  All quoted probabilities and Fisher matrices are
recomputed from the unpruned saved effects.

Every effect update in this loop is a certified convex optimization.  Unless
the pricing problem is exhausted over the complete score body, however, the
outer iteration remains a restricted/local search rather than a global
unrestricted-POVM certificate.  The saved deliverables are the effects,
probabilities, derivatives, scores, CFIM, optimization history, and all
validation residuals; they are sufficient to reproduce both the likelihood
and the receiver.

\subsection{Permutation reduction and physical receiver representation}

The stored temporal records are allowed to interfere coherently before any
one of them is detected.  Since \(\sigma\), \(\dot\sigma_e\), and
\(F_e^{(n_s)}\) are invariant under copy permutations, every effect can first
be twirled and decomposed as
\begin{equation}
 M_y=\bigoplus_\lambda M_y^{(\lambda)}\otimes I_{m_\lambda}.
 \label{eq:schur_effect}
\end{equation}
For \(N=4\) and \(n_s=2\), the Schur transform separates the
\(4\times4=16\) two-record amplitudes into a ten-dimensional
exchange-symmetric sector, spanned by \(\lvert ii\rangle\) and
\((\lvert ij\rangle+\lvert ji\rangle)/\sqrt2\), and a six-dimensional
exchange-antisymmetric sector, spanned by
\((\lvert ij\rangle-\lvert ji\rangle)/\sqrt2\).  These are exchange
symmetries of two addressable temporal records, not telescope labels or
accepted and rejected events.  The product state \(\rho_0^{\otimes2}\)
generally has weight in both sectors, and both carry phase response.  At
\(n_s=3\), the reduced dimensions are \(20\oplus20\oplus4\), with
multiplicities \(1,2,1\).  The symmetry reduction therefore preserves all
available information while making the optimization tractable.

When a sharp receiver is preferred, parameterize one basis in every Schur
block by a unitary \(U_\lambda\).  Then
\begin{align}
 M_{\lambda a}&=U_\lambda|a\rangle\!\langle a|U_\lambda^\dagger
 \otimes I_{m_\lambda},\notag\\
 p_{\lambda a}&=m_\lambda
 [U_\lambda^\dagger\sigma_\lambda U_\lambda]_{aa},\qquad
 d_{\lambda a,e}=m_\lambda
 [U_\lambda^\dagger\dot\sigma_{\lambda,e}U_\lambda]_{aa}.
 \label{eq:block_pvm_statistics}
\end{align}
Each column of \(U_\lambda\) is one collective output mode: a coherent
superposition of station patterns from several records within a fixed
exchange sector.  The optimizer chooses their relative amplitudes and phases
so that the resulting count probabilities respond to all estimated edge
phases as evenly and independently as possible.  Minimizing
\(\Tr(WJ^{-1})\) strongly penalizes a blind or weak direction instead of
rewarding only the most sensitive port.  The PVM remains deterministic and
complete: every two-copy block exits through one of the 16 ports, so no
postselection probability multiplies its Fisher information.

The objective and its gradient are therefore evaluated without forming the
full \(N^{n_s}\)-dimensional effects.  Starting from independent Haar-random
block bases, we take a Riemannian descent step, use an Armijo line search,
and restore unitarity by QR retraction.  Multiple restarts are ranked by
\(\Tr(WJ^{-1})\), and the best basis is exported with
\(\bm p,\bm d,J\) and the residuals in Eq.~\eqref{eq:numerical_residuals}.
This direct PVM branch is nonconvex but especially transparent; it is the
branch used for the explicit \(n_s=2\) receiver below.  To enlarge the search
from projective measurements to overcomplete POVMs, we replace each
\(d_\lambda\times d_\lambda\) block unitary by a Parseval frame
\(V_\lambda\in\mathbb C^{d_\lambda\times qd_\lambda}\), satisfying
\(V_\lambda V_\lambda^\dagger=I_{d_\lambda}\).  Its columns
\(\{|v_{\lambda a}\rangle\}\) define the effects
\(|v_{\lambda a}\rangle\!\langle v_{\lambda a}|\otimes I_{m_\lambda}\).
The \(q=2\) searches reported below therefore have \(32\) outcomes at
\(n_s=2\) and \(88\) outcomes at \(n_s=3\).  We optimize the same A-risk and
restore the Parseval constraint by a row-isometry retraction after each step.

The numerical output also specifies a circuit.  Apply the Schur transform,
the block-controlled basis change \(U_\lambda^\dagger\), and then resolve the
block label and output port.  For \(N=4,n_s=2\), this is a 16-dimensional
joint unitary followed by two identical four-outcome computational-basis
detections, exactly matching the joint-processing plus single-copy-readout
architecture of the main text.  A general SDP POVM is compiled by factorizing
\(M_y=R_y^\dagger R_y\), stacking the \(R_y\) into a Naimark isometry,
completing that isometry to a unitary, and measuring its outcome register.
The optimization therefore returns either the intermediate unitary directly
(PVM route) or the data needed to synthesize its dilated version (POVM route).

\subsection{\texorpdfstring{Worked example: \(N=4\) and \(n_s=2\)}
{Worked example: N=4 and ns=2}}

We use the full-rank complex working state
\begin{equation}
 \rho_0=\frac14
 \begin{pmatrix}
 1 & \frac12 & \frac12e^{i\pi/12} & \frac12e^{i\pi/6}\\
 \frac12 & 1 & \frac12e^{i\pi/6} & \frac12e^{i\pi/12}\\
 \frac12e^{-i\pi/12} & \frac12e^{-i\pi/6} & 1 & \frac12e^{-i\pi/6}\\
 \frac12e^{-i\pi/6} & \frac12e^{-i\pi/12} & \frac12e^{i\pi/6} & 1
 \end{pmatrix}.
 \label{eq:n4_state}
\end{equation}
Thus \(|g_{ij}|=0.5\) on every edge, whereas the edge-phase vector in the
order \((12,13,14,23,24,34)\) is
\((0,\pi/12,\pi/6,\pi/6,\pi/12,-\pi/6)\). The four triangle phases are
\((\pi/12,-\pi/12,-\pi/4,-\pi/12)\), so station phases cannot transform
Eq.~\eqref{eq:n4_state} into a real matrix.
We estimate all six edge phases while holding the visibility magnitudes
fixed. The locally phase-matched repetitive uniform-edge-first POVM has
\begin{equation}
 J_{\rep}=I_6/24,\qquad \mathcal R_{\rep}=\Tr J_{\rep}^{-1}=144.
 \label{eq:n4_repetitive}
\end{equation}

For a completely explicit receiver, we use the Schur-projective route in
Eq.~\eqref{eq:block_pvm_statistics}. Positivity and completeness then hold
analytically throughout the search. This is a constructive optimization
within the projective-Schur subclass, rather than a proof of globally optimal
performance over unrestricted POVMs.

To specify the collective receiver explicitly, we order the product basis as
\(\mathcal B_{\rm prod}=(|11\rangle,|12\rangle,\ldots,|44\rangle)\) and use
the two Schur-sector bases
\begin{align}
 \mathcal B_+={}&
 \bigl(|11\rangle,|22\rangle,|33\rangle,|44\rangle,\notag\\[-2pt]
 &\hspace{1.1em}|12\rangle_+,|13\rangle_+,|14\rangle_+,
 |23\rangle_+,|24\rangle_+,|34\rangle_+\bigr),\notag\\
 \mathcal B_-={}&
 \bigl(|12\rangle_-,|13\rangle_-,|14\rangle_-,\notag\\[-2pt]
 &\hspace{1.1em}|23\rangle_-,|24\rangle_-,|34\rangle_-\bigr),
 \label{eq:n4_schur_bases}
\end{align}
where
\(|ij\rangle_\pm=(|ij\rangle\pm|ji\rangle)/\sqrt2\).
Let \(S_+\) and \(S_-\) denote the matrices whose columns are the vectors in
\(\mathcal B_+\) and \(\mathcal B_-\), respectively, written in
\(\mathcal B_{\rm prod}\), and define the \(16\times16\) Schur matrix
\(S=(S_+,S_-)\). The optimized variables are a \(10\times10\) unitary
\(U_+\) and a \(6\times6\) unitary \(U_-\). In the original two-copy basis,
the complete collective-mode matrix and the unitary actually applied before
detection are therefore
\begin{equation}
 U_{\rm coll}
 =S\begin{pmatrix}U_+&0\\0&U_-\end{pmatrix},
 \qquad
 V_{\rm read}=U_{\rm coll}^\dagger
 =\begin{pmatrix}U_+^\dagger&0\\0&U_-^\dagger\end{pmatrix}S^\dagger .
 \label{eq:n4_collective_unitary}
\end{equation}
Thus the first operation \(S^\dagger\) resolves exchange symmetry without
measuring it, and the controlled rotations \(U_\pm^\dagger\) coherently
recombine all station-pair amplitudes within the corresponding sector.

After \(V_{\rm read}\), two four-output computational-basis detectors report
\((a,b)\in\{1,\ldots,4\}^2\), with
\(y=4(a-1)+b\). Pulling this terminal detection back through the unitary
gives the input-side PVM
\begin{align}
 M_{ab}&=|\psi_{ab}\rangle\!\langle\psi_{ab}|,\notag\\
 |\psi_{ab}\rangle&=U_{\rm coll}|a,b\rangle_{\rm out},\notag\\
 p_{ab}&=\langle\psi_{ab}|\rho_0^{\otimes2}|\psi_{ab}\rangle .
 \label{eq:n4_terminal_projectors}
\end{align}
Equivalently, the first ten \(|\psi_{ab}\rangle\)'s are the columns
\(S_+U_+|r\rangle\), and the remaining six are
\(S_-U_-|r\rangle\). Equation~\eqref{eq:n4_terminal_projectors} makes clear
that the two final detectors are local only in the \emph{output} ports: in
the original record basis, every click projects onto a generally entangled
two-copy state.

For example, after fixing the ordering in Eq.~\eqref{eq:n4_schur_bases}, the
coordinate vectors of the first symmetric and antisymmetric projection states
of the optimized receiver are
\begin{equation}
 [\psi_{+,1}]_{\mathcal B_+}=
 \begin{pmatrix}
 -0.4633-0.0841i\\
 -0.4981-0.0252i\\
 -0.1908+0.4309i\\
 -0.4728+0.2412i\\
 \phantom{-}0.0641-0.0183i\\
 -0.0366+0.0387i\\
 \phantom{-}0.0448-0.0350i\\
 \phantom{-}0.0661-0.0090i\\
 \phantom{-}0.0773-0.0265i\\
 \phantom{-}0.0572-0.0289i
 \end{pmatrix},
 \label{eq:n4_example_symmetric_port}
\end{equation}
and
\begin{equation}
 [\psi_{-,1}]_{\mathcal B_-}=
 \begin{pmatrix}
 -0.4033+0.0708i\\
 -0.0747-0.0293i\\
 \phantom{-}0.2224+0.5472i\\
 -0.2346-0.4638i\\
 \phantom{-}0.0859-0.0965i\\
 \phantom{-}0.3098-0.3070i
 \end{pmatrix}.
 \label{eq:n4_example_antisymmetric_port}
\end{equation}
At the working point in Eq.~\eqref{eq:n4_state}, these outcomes have
probabilities \(p_{+,1}=0.07218\) and \(p_{-,1}=0.04278\), respectively.
The other fourteen projection states are obtained from the remaining columns
of \(U_+\) and \(U_-\) in exactly the same way. The complete complex matrices
used here are supplied with the Supplemental Data, so
Eqs.~\eqref{eq:n4_collective_unitary} and
\eqref{eq:n4_terminal_projectors} reproduce all sixteen effects without any
additional optimization.

\begin{figure*}[t]
 \includegraphics[width=0.99\textwidth]
 {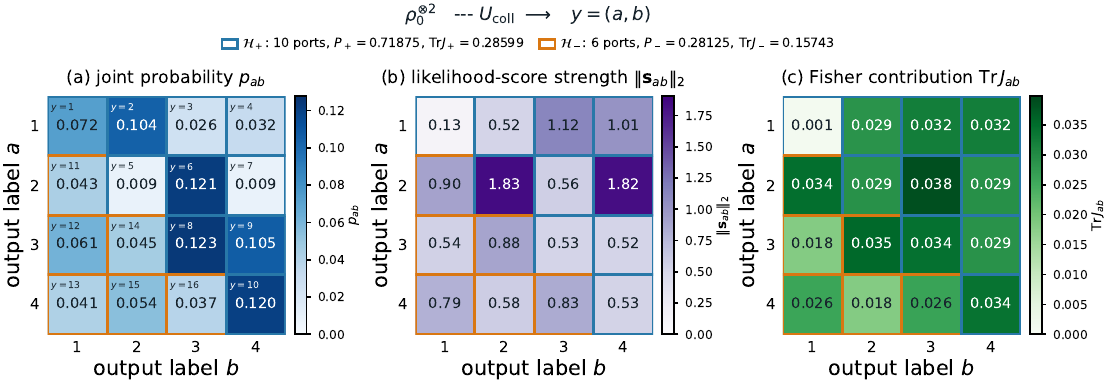}
 \caption{\label{fig:pvm_outcomes}
 Complete outcome atlas for the optimized \(N=4,n_s=2\) collective PVM at
 Eq.~\eqref{eq:n4_state}. Each cell is one resolved port after the collective
 unitary \(U_{\rm coll}\); the pair \((a,b)\) is an output label, not a copy
 or telescope identity. The three panels show its Born probability \(p_y\),
 local-score norm
 \(\|\bm s_y\|_2=\|\nabla_{\bm h}\ln p_y\|_2\), and total Fisher contribution
 \(\Tr J_y=p_y\|\bm s_y\|_2^2\), respectively. Blue and orange boundaries
 identify the exchange-symmetric and exchange-antisymmetric sectors; all 16
 outcomes enter the likelihood. Comparing the panels shows why outcome
 probability alone does not determine the information it carries.}
\end{figure*}

Operationally, this example asks which 16 orthogonal joint modes should be
counted after two four-station records have been stored. The repetitive
edge-first receiver selects a baseline and collapses each record separately.
The collective receiver instead preserves both records, coherently mixes
their station-pair amplitudes through \(U_{\rm coll}\), and performs the
strong readout only afterward. A resolved port therefore need not correspond
to one baseline: its probability can respond simultaneously to several edge
phases.

The PVM contains ten symmetric and six antisymmetric rank-one output ports.
Their total sector weights are
\begin{equation}
 P_\pm=\frac{1\pm\Tr(\rho_0^2)}{2}
 =(0.71875,0.28125).
 \label{eq:n4_sector_weights}
\end{equation}
Because the visibility magnitudes are fixed in this phase-estimation example,
\(\Tr(\rho_0^2)\) is independent of the six phases. The sector label alone
therefore carries no phase information; the gain comes from the
phase-dependent redistribution of probability among ports within each
sector.

Figure~\ref{fig:pvm_outcomes} resolves this mechanism outcome by outcome. A
port with probability \(p_y\) and local score \(\bm s_y\) contributes
\(p_y\bm s_y\bm s_y^{\mathsf T}\) to the CFIM. Probability and information
are therefore not interchangeable: a less frequent port can be highly
informative when its probability changes rapidly with the phases. A generic
collective-port score also has several nonzero components, so one click
constrains several edge phases rather than identifying one edge. The
optimized basis arranges these score vectors to span the six-dimensional
phase space without a weak direction.

The resulting CFIM has eigenvalues between \(0.0580\) and \(0.0895\), all
above the repetitive value \(1/24\), and gives \(\mathcal R=84.4204\).
Thus the lower risk is not obtained by sacrificing one phase direction to
improve another; it reflects a more balanced sensitivity across the full
six-dimensional phase space. Independent reconstruction of the full-space
effects verifies completeness, copy-permutation symmetry, and the CFIM to
machine precision.

\subsection{Finite-copy comparison and asymptotic limit}

We compare the receivers edge by edge without pretending that the other five
phases are known.  For the coordinate vector \(\bm e_e\) of physical edge
\(e\), define the nuisance-aware directional Fisher gain
\begin{equation}
 G_e(J)\equiv
 \frac{[\bm e_e^{\mathsf T}J^{-1}\bm e_e]^{-1}}
 {[\bm e_e^{\mathsf T}J_{\rep}^{-1}\bm e_e]^{-1}}
 =\frac{24}{[J^{-1}]_{ee}}.
 \label{eq:edge_fisher_gain}
\end{equation}
Unlike the ratio of diagonal CFIM entries, Eq.~\eqref{eq:edge_fisher_gain}
retains the covariance penalty from estimating all six edge phases
simultaneously.  Figure~\ref{fig:finite_edge_gains} reports every \(G_e\) and
their arithmetic mean, while Table~\ref{tab:finite_comparison} collects the
same comparison together with the global A-risk.

The overcomplete \(q=2\) receivers replace one orthogonal basis in each Schur
sector by a larger Parseval frame.  Physically, the extra output modes sample
more directions of the six-dimensional score space, while the third coherent
record at \(n_s=3\) supplies additional interfering multi-record pathways.
Figure~\ref{fig:finite_edge_gains} and
Table~\ref{tab:finite_comparison} show the resulting progression from the
two-copy PVM to the two- and three-copy POVMs and finally to the asymptotic
limit.  The finite-copy entries are the best validated local optima within
their stated receiver classes, rather than certificates of globally optimal
unrestricted POVMs.  At the complex working point
Eq.~\eqref{eq:n4_state}, the weak-commutator matrix has rank six, so the QFIM
itself is not an attainable covariance benchmark; the appropriate asymptotic
reference is the numerical unit-weight Holevo limit.

\begin{figure}[tb]
 \includegraphics[width=0.92\columnwidth]{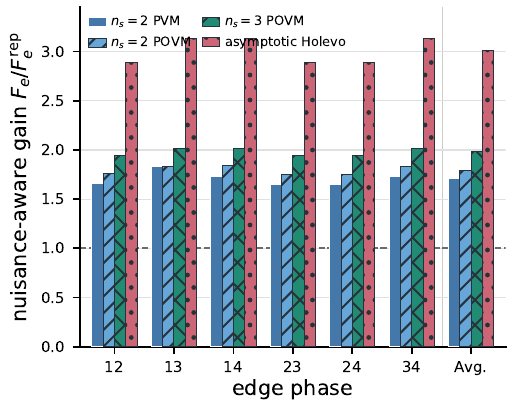}
 \caption{\label{fig:finite_edge_gains}
 Nuisance-aware Fisher gain for each physical edge phase at the \(N=4\)
 working point in Eq.~\eqref{eq:n4_state}; ``Avg.'' is the arithmetic mean
 over the six edges.  Each finite-copy bar uses the full inverse CFIM through
 Eq.~\eqref{eq:edge_fisher_gain}, so the other five phases remain nuisance
 parameters.  The final group is the numerical unit-weight Holevo limit.
 }
\end{figure}

\begin{table}[tb]
 \caption{\label{tab:finite_comparison}
 Finite-copy performance at Eq.~\eqref{eq:n4_state}.  Here \(V=J^{-1}\) for
 the explicit finite-copy receivers, \(G_A=144/\Tr V\), and
 \(\overline G_e\) averages the six directional gains.
 }
 \small
 \resizebox{\columnwidth}{!}{%
 \begin{tabular}{lccccc}
 \toprule
 receiver & outcomes & \(\Tr V\) & \(G_A\) & \(\overline G_e\) & range of \(G_e\)\\
 \midrule
 repetitive edge-first & \(12\) per copy & 144.000 & 1.000 & 1.000 & 1.000--1.000\\
 \(n_s=2\) Schur PVM & 16 & 84.420 & 1.706 & 1.708 & 1.647--1.834\\
 \(n_s=2\), \(q=2\) Schur POVM & 32 & 80.267 & 1.794 & 1.795 & 1.751--1.838\\
 \(n_s=3\), \(q=2\) Schur POVM & 88 & 72.703 & 1.981 & 1.981 & 1.948--2.015\\
 numerical asymptotic Holevo limit & --- & 47.892 & 3.007 & 3.012 & 2.889--3.135\\
 \bottomrule
 \end{tabular}}
\end{table}

\section{Loss, background, and atmospheric-noise model}
\label{sec:noise_model}

\subsection{Open-system description of loss and background}

The red dashed box in Fig.~1 of the main text contains optical collection,
transport, memory encoding, joint processing, and detection. Before the final
POVM, all imperfect components can be regarded as an open-system channel
\(\mathcal E\) acting on the complete weak thermal state, including its vacuum
and multiphoton sectors. A standard dilation of local attenuation couples the
received mode at station \(i\) to an environment mode \(\hat e_i\),
\begin{equation}
 \begin{aligned}
 \hat a_{i,\mathrm{det}}
 &=\sqrt{\eta_i}\,\hat a_{i,\mathrm{src}}
   +\sqrt{1-\eta_i}\,\hat e_i,\\
 \mathcal E(\rho)
 &=\Tr_E\!\left\{U_{\Lambda}
   [\rho\otimes\tau_E]U_{\Lambda}^{\dagger}\right\}.
 \end{aligned}
 \label{eq:thermal_loss_channel}
\end{equation}
where \(\eta_i\) is the net efficiency and \(\tau_E\) describes the light and
device noise coupled into the detected mode.
Equation~\eqref{eq:thermal_loss_channel} is a completely positive
trace-preserving channel; loss transfers excitations to the unobserved
environment and correspondingly increases the vacuum weight, whereas a
populated environment injects background photons into the receiver. This
beamsplitter dilation is the usual thermal-loss model for a bosonic
communication channel~\cite{HolevoWerner2001,Weedbrook2012}.

At the level of the first-order coherence matrix
\(\Gamma_{ij}=\Tr(\rho\hat a_i^\dagger\hat a_j)\), the same open-system
evolution has the general affine form
\begin{equation}
 \Gamma_{\mathrm{det}}
 =\Lambda^{1/2}\Gamma_{\mathrm{src}}\Lambda^{1/2}+\rho^L.
 \label{eq:effective_loss_background}
\end{equation}
The positive transfer operator \(\Lambda\) describes imperfections in signal
collection, transport, storage, and readout: its eigenvalues give the retained
signal efficiencies, while nondiagonal elements can represent coherent mixing
between received modes. The positive additive term \(\rho^L\) is the
occupation matrix of photons coupled into the receiver from the environment,
including sky, memory, transport, and detector backgrounds. It is not an
additional normalized density operator; its diagonal entries are background
occupation numbers, whereas its off-diagonal entries would describe mutually
coherent background fields.

For the telescope network considered in the main text, we specialize this
general model to
\begin{equation}
 \Lambda=\operatorname{diag}(\eta_1,\ldots,\eta_N),
 \qquad
 \rho^L=\operatorname{diag}(\varepsilon_1,\ldots,\varepsilon_N).
 \label{eq:diagonal_loss_background}
\end{equation}
The diagonal form of \(\Lambda\) follows from the station-resolved
architecture: aperture collection, spatial-mode injection, memory write/read,
fiber transmission, and detection occur along separately addressable optical
paths. These processes attenuate the amplitude carried by a station but do
not coherently transfer it to another station label. Their efficiencies can
therefore be combined into one scalar \(\eta_i\) for each path, giving
\(\Gamma_{ij}\mapsto\sqrt{\eta_i\eta_j}\Gamma_{ij}\).

The diagonal form of \(\rho^L\) has a different physical origin. The dominant
sky, memory, fiber, and detector backgrounds are generated locally and have no
stable phase relation between remote stations; ensemble averaging therefore
removes their cross-station coherences, leaving the mean local occupations
\(\varepsilon_i\). Coherent cross talk or a spatially correlated background
would require nondiagonal \(\Lambda\) or \(\rho^L\), respectively, and can be
retained within the general form of Eq.~\eqref{eq:effective_loss_background}.
Neither effect is expected to dominate the station-separated configuration
modeled here.

In the weak-light regime, dividing \(\Gamma_{\mathrm{det}}\) by its trace gives
the normalized one-photon matrix \(\rho^{(1)}\) used in the main-text
estimation model. Thus the single-photon description is the leading
informative sector of the full CPTP evolution, not an assumption that loss and
background act only after one-photon postselection.

Finally, the channel and the terminal detector may be combined into one
effective measurement. If \(\{\Pi_x\}\) is the POVM after the lossy memory
network, then the input-state probabilities are
\(p_x=\Tr[\mathcal E(\rho)\Pi_x]
=\Tr[\rho\,\mathcal E^\dagger(\Pi_x)]\). Hence the entire red box is
equivalently described by the effective POVM
\(\Pi_x^{\rm eff}=\mathcal E^\dagger(\Pi_x)\), including no-click and failure
outcomes. A trace-decreasing map appears only if these outcomes are discarded
and one conditions on successful detection.

\subsection{Classical disturbances and the use of closure phase}

In addition to quantum loss and background, a ground-based optical array is
affected by three principal classes of classical disturbance
~\cite{Monnier2003,Monnier2007,Tatulli2007}:
\begin{itemize}
 \item \emph{Station piston.} Atmospheric path-length fluctuations and
 instrumental delay errors add an approximately uniform phase \(\delta_i\)
 across the selected mode of station \(i\). They transform a baseline
 coherence as
 \(\Gamma_{ij}\mapsto
 e^{\mathrm i(\delta_i-\delta_j)}\Gamma_{ij}\), directly masking the source
 visibility phase.

 \item \emph{Higher-order wavefront distortion.} Turbulence also produces a
 phase profile that varies across a telescope pupil. Before correction and
 projection into a common spatial mode, the field cannot be represented by a
 single station phase. The residual distortion lowers the Strehl ratio and
 the coherent coupling or fringe contrast rather than acting as a pure
 piston.

 \item \emph{Amplitude fluctuation.} Scintillation, transparency variations,
 pointing and coupling drift, and gain instability modulate the received
 photon rate. For the meter-class apertures considered here, aperture and
 temporal averaging suppress atmospheric scintillation, while simultaneous
 photometric monitoring calibrates slow throughput
 changes~\cite{Kornilov2012}. We therefore treat the residual amplitude noise
 as subdominant and absorb it into the calibrated efficiencies and data
 covariance.
\end{itemize}

The latter two classes admit effective station-local mitigation. Adaptive
optics (AO) measures the spatially varying wavefront and corrects it with a
deformable mirror; spatial filtering then converts the residual mismatch into
a measurable coupling loss. Modern high-order systems operating at visible
and near-visible wavelengths have demonstrated residual control at the level
of roughly one tenth of a wavelength, and in favorable regimes below that
scale~\cite{Landman2024MagAOX}. The remaining loss of coherent throughput is
therefore represented by \(\eta_i\) in
Eq.~\eqref{eq:diagonal_loss_background} rather than by an additional unknown
baseline phase.

A spatially uniform piston is qualitatively different: a local AO sensor is
insensitive to the absolute phase shared by the whole pupil and therefore
cannot establish the differential phase between separated telescopes. It
must instead be measured by fringe tracking on the science target or on a
sufficiently bright nearby natural reference. Existing fringe trackers can
strongly suppress differential optical-path fluctuations, but their operation
depends on the flux and angular proximity of a suitable
reference~\cite{Lacour2019FringeTracker}. Many faint or isolated targets do
not provide such a reference within the relevant atmospheric patch, so
uncalibrated or residual piston remains the least reliably removable phase
disturbance.

This motivates retaining piston explicitly and eliminating it with closure
quantities. Writing the measured baseline phase as
\(\phi_{ij}^{\rm obs}=\phi_{ij}+\delta_i-\delta_j\), an oriented triangle obeys
\begin{equation}
 \Phi_{ijk}^{\rm obs}
 =\phi_{ij}^{\rm obs}+\phi_{jk}^{\rm obs}+\phi_{ki}^{\rm obs}
 =\phi_{ij}+\phi_{jk}+\phi_{ki}\equiv\Phi_{ijk}.
 \label{eq:noise_model_closure}
\end{equation}
The cancellation is exact for arbitrary station-local pistons and does not
require them to remain stable over the total observing time, provided that
the constituent baselines are acquired within a common phase-connected
sample. Closure does not remove genuinely baseline-dependent or nonclosing
instrumental errors; those must still be calibrated or included in the
covariance model. Under the station-local model relevant here, however,
piston is the dominant untracked phase nuisance after AO and photometric
calibration. We consequently retain all visibility-amplitude information but
project the phase data onto the piston-invariant closure space. The following
section implements this projection and its nuisance elimination explicitly in
the RML likelihood.

\section{Regularized maximum-likelihood imaging test}
\label{sec:rml}
\subsection{Data model and receiver covariances}

The imaging test keeps the Fourier coverage and reconstruction algorithm fixed
while changing only the receiver-induced data covariance.  For wavelength
\(\lambda\), projected baseline \(\bm b_{ij}(t)\), and pixelized trial image
\(I_{\lambda,p}\), the model visibility is
\begin{equation}
 g_{ij}^{\rm mod}(\lambda,t)=
 \frac{\sum_pI_{\lambda,p}
 e^{-2\pi i\bm b_{ij}(t)\cdot\bm\theta_p/\lambda}}
 {\sum_pI_{\lambda,p}}.
 \label{eq:rml_visibility}
\end{equation}

Ground-based phase measurements are corrupted by rapidly varying,
station-dependent atmospheric pistons.  Closure phases remove these local
shifts and are therefore the phase observables supplied to the RML
reconstruction.  Orient the \(E=N(N-1)/2\) baselines once and collect their
measured phases in \(\bm\phi^{\rm obs}\).  After removing the unobservable
global piston, the identifiable local parametrization is
\begin{equation}
 \bm\phi^{\rm obs}=Q\bm\Phi+K\bm\delta ,
 \qquad Q^{\mathsf T}K=0.
 \label{eq:rml_edge_decomposition}
\end{equation}
Here the \(N-1\) columns of \(K\) span the station-incidence, or cut, space,
the \(C=E-(N-1)\) columns of \(Q\) span the orthogonal cycle space, and
\(\bm\delta\) contains the unknown station pistons.  The vector \(\bm\Phi\)
therefore contains all phase combinations that are invariant under
station-based atmospheric shifts~\cite{Jennison1958,Monnier2003}.

This coordinate change does not make \(\bm\delta\) known.  Partitioning the
receiver CFIM in the coordinates \((\bm\Phi,\bm\delta)\), the usable closure
information is obtained by eliminating the piston nuisance through the Schur
complement
\begin{equation}
 J_{\Phi}
 =J_{\Phi\Phi}
 -J_{\Phi\delta}J_{\delta\delta}^{+}J_{\delta\Phi}.
 \label{eq:rml_full_schur}
\end{equation}
The subtracted term is precisely the apparent closure response that can be
reproduced by changing the unknown station pistons.  For the six-station
array, \(E=15\), \(N-1=5\), and \(C=10\), so this elimination concentrates
the phase information from the 15-dimensional edge space into ten
atmosphere-invariant coordinates.  The implemented likelihood retains these
coordinates together with all 15 visibility amplitudes and applies the same
nuisance elimination to the combined amplitude--closure block before
inversion.  This gives the \(25\times25\) covariance used below, including
its amplitude--closure cross block.  The loop-only gain in
Fig.~\ref{fig:rml_stats}(b) applies one further Schur complement that treats
the amplitudes as nuisance parameters.

The model value of each closure phase is
\begin{equation}
 \Phi_{ijk}^{\rm mod}=\arg g_{ij}^{\rm mod}
 +\arg g_{jk}^{\rm mod}+\arg g_{ki}^{\rm mod}.
 \label{eq:rml_closure}
\end{equation}
The chosen balanced basis is
\[
 \{123,124,125,134,136,245,256,346,356,456\}.
\]
The covariance is transformed from the orthonormal \(Q\) coordinates to these
ten physically transparent triangle coordinates before reconstruction; no
station-piston phase is given to the optimizer.

The reduction also permits integration beyond the atmospheric coherence
time.  In each short atmospheric cell \(b\), \(\bm\delta_b\) may be entirely
different, but \(Q^{\mathsf T}K\bm\delta_b=0\) identically.  Every cell
therefore probes the same source closure coordinates, and the independent
information contributions add,
\begin{equation}
 J_{\Phi}^{\rm tot}=\sum_b J_{\Phi,b}.
 \label{eq:rml_long_time_addition}
\end{equation}
This is a statistical accumulation of gauge-invariant data, not a coherent
average of raw edge phases or optical fields over the full observing time.
It assumes simultaneous baseline sampling and station-based phase errors;
non-closing instrumental terms must instead be calibrated or included as
additional nuisance parameters.

For each wavelength--epoch sample, the residual vector contains the fifteen
visibility amplitudes and ten closure phases.  The likelihood uses the full
receiver-specific \(25\times25\) covariance, including its amplitude--closure
cross block,
\begin{equation}
 \chi_s^2(I)=\sum_{\lambda,t}
 \bm r_{s,\lambda t}(I)^{\mathsf T}
 \Sigma_{s,\lambda t}^{+}\bm r_{s,\lambda t}(I),
 \label{eq:rml_likelihood}
\end{equation}
Here the closure residual is evaluated in the
real-valued local tangent coordinates of the calibrated POVM working point,
rather than by independently wrapping reconstructed edge phases.  Half of
the detected photons is assigned to the amplitude branch and half to the phase
branch for all three strategies.  The synthetic draws are paired as
\begin{equation}
 \bm d_{\lambda t}^{(s)}
 =\bm d_{\lambda t}^{\rm true}+L_{\lambda t}^{(s)}\bm z_{\lambda t},
 \qquad
 L_{\lambda t}^{(s)}L_{\lambda t}^{(s)\mathsf T}
 =\Sigma_{\lambda t}^{(s)},
 \label{eq:paired_draw}
\end{equation}
where the same standard-normal vector \(\bm z_{\lambda t}\) is used for every
strategy \(s\), while \(L_{\lambda t}^{(s)}\) is receiver dependent.

We compare (i) uniform edge-first readout, (ii) the
locally optimized one-copy POVM, and (iii) the joint receiver induced by that
same POVM for a finite coherent block \(n_s=3\).  The one-copy POVM is optimized
independently at every wavelength--epoch working point.  For the joint
receiver, an overcomplete three-copy POVM is additionally optimized at the
near-transit working point of each 10-nm channel, whose QFI-weighted Fisher
retention relative to the asymptotic receiver ranges from \(0.306\) to
\(0.307\) across the ten channels and rescales only the corresponding channel;
the same calibrated factor is then reused across noise seeds.  

\subsection{Likelihood-matched RML reconstruction}

For each wavelength channel, we optimize the \(40\times40\) real-space pixel
intensities directly, subject to nonnegativity and unit total flux.  For
receiver \(s\), the reconstruction minimizes
\begin{equation}
 \begin{aligned}
 \mathcal L_s(I)={}&\chi_s^2(I)
 +w_{\rm prior}\|I-I_0\|_2^2
 +w_{\rm TV}{\rm TV}(I)\\
 &+w_{\rm ent}\sum_p I_p\ln\frac{I_p}{I_{0,p}}.
 \end{aligned}
 \label{eq:rml_objective}
\end{equation}
Here \(I_0\) is a broad Gaussian reference image.  The quadratic prior keeps
the large-scale morphology near \(I_0\), the total-variation term suppresses
pixel-scale oscillations while retaining sharp structures, and the relative
entropy discourages fragmented low-intensity features.  No translation-gauge,
compact-core, single-centre, or other morphology-specific penalty is imposed.
The common weights are
\begin{equation}
 (w_{\rm prior},w_{\rm TV},w_{\rm ent})
 =(0.01,0.01,0.005).
 \label{eq:rml_baseline_weights}
\end{equation}
These weights, the positivity and flux constraints, and the \(40\times40\)
optimization grid are fixed before the ensemble run and are identical for
every receiver, wavelength, and noise seed.

The candidate is selected using only its fit to the measured data.  After
whitening the mixed amplitude--closure residual with the full covariance in
Eq.~\eqref{eq:rml_likelihood}, we define
\begin{equation}
 \overline{\chi}^{\,2}_{\rm amp}
 =\frac{1}{N_{\rm amp}}\sum_q z_{{\rm amp},q}^{2},
 \qquad
 \overline{\chi}^{\,2}_{\Phi}
 =\frac{1}{N_{\Phi}}\sum_q z_{\Phi,q}^{2}\,,
 \label{eq:rml_component_chi2}
\end{equation}
where the scale factor \(\chi\) quantifies the consistency between the
reconstructed real-space image and the collected data samples at the
corresponding Fourier components.  In a standard RML pipeline, a reconstruction
is generally considered acceptable when \(\chi^2\leq 3\).

For each receiver and wavelength channel, we test the dirty-image and
Gaussian-prior starts.  We retain candidates in the common discrepancy window
\(0.75\leq\overline\chi^2_{\alpha}\leq1.00\) and select the one minimizing
\begin{equation}
 D=\max_{\alpha\in\{{\rm amp},\Phi\}}
 \left|\ln\frac{\overline{\chi}^{\,2}_{\alpha}}{0.85}\right|.
 \label{eq:rml_discrepancy}
\end{equation}
This discrepancy rule is receiver blind and does not use the true image or
any image-correlation metric.

All candidates use Adam with 2600 base iterations and learning rate
\(10^{-2}\).  If the componentwise discrepancy window is not reached, longer
passes use iteration multipliers \((1.0,1.8,3.0,4.6)\) and learning-rate
multipliers \((1.0,0.75,0.55,0.40)\).  The candidate set and schedule are
identical for all three receivers.

The simulated observation covers \(600\)--\(700\,\mathrm{nm}\) in ten
\(10\,\mathrm{nm}\) channels and 36 Earth-rotation epochs separated by
15 min.  Each sample has \(100\,\mathrm{ms}\) effective integration, giving
an \(8.75\,\mathrm{h}\) Fourier-coverage span.  The six-station array uses
full-aperture collection from the three central facilities and three
\(6\,\mathrm{m}\) remote stations; its maximum separation is approximately
\(10.88\,\mathrm{km}\).  We assume 2\% local photon-collection efficiency,
\(0.2\,\mathrm{dB/km}\) fiber attenuation, and background occupancy
\(10^{-9}\).  The source is an NGC~4151-like core--disk--BLR model with
\(m_V\simeq11.9\), a BLR radius of \(72\,\mu\mathrm{as}\), and a width of
\(12\,\mu\mathrm{as}\).  Each spectral channel is reconstructed independently,
and the displayed broadband image is their photon-weighted stack, without any
post-reconstruction translation or image registration.

\subsection{Paired-seed imaging statistics}

We reran the complete reconstruction pipeline for the twelve paired noise
seeds.  For a given seed, the same underlying standard-normal noise
realization is used for all three receivers and is mapped through their
respective covariance matrices.  Thus the receiver comparisons are paired
while retaining the correct receiver-dependent noise.

The top panel of figure~4 in the main text shows the BLR-annulus and all-pixel correlations.
The circles denote individual paired seeds, whereas the bars and error bars
give the mean and standard error.  The representative reconstruction shown in
Fig.~3 of the main text uses seed \(20260536\), chosen as the seed closest to the
twelve-seed mean over the six standardized image-correlation coordinates
corresponding to the three receivers and two image regions.  This choice is
used only for visualization; all statistical conclusions below include the
full twelve-seed ensemble.

Table~\ref{tab:rml_chi2} reports the corresponding componentwise residuals as
a numerical convergence diagnostic.  Across the three receivers, the mean
amplitude residuals span only \(0.834\)--\(0.871\), while the closure residuals
span \(0.885\)--\(0.922\).  The maximum-to-minimum ratios are \(1.044\) and
\(1.041\), respectively, well within the prescribed 10\% matching tolerance.
The comparison therefore probes differently informative data at comparable
fitting quality, rather than rewarding one receiver through systematic
underfitting or overfitting.

\begin{table}[t]
\caption{\label{tab:rml_chi2}
Componentwise reduced residuals for the three common-prior,
likelihood-matched RML reconstructions.  Entries are the mean \(\pm\) standard
error over twelve paired noise seeds.}
\centering
\small
\begin{ruledtabular}
\begin{tabular}{lcc}
Receiver
& \(\overline{\chi}^{\,2}_{\rm amp}\)
& \(\overline{\chi}^{\,2}_{\Phi}\)\\
\hline
Uniform edge-first
& \(0.834\pm0.005\) & \(0.918\pm0.006\)\\
Optimal one-copy POVM
& \(0.850\pm0.005\) & \(0.922\pm0.007\)\\
Joint POVM, \(n_s=3\)
& \(0.871\pm0.005\) & \(0.885\pm0.006\)\\
\end{tabular}
\end{ruledtabular}
\end{table}

The BLR-annulus correlations are \(0.686\pm0.007\), \(0.726\pm0.007\), and
\(0.810\pm0.004\) for the uniform edge-first, optimal one-copy, and
finite-\(n_s=3\) joint receivers, respectively.  The corresponding all-pixel
correlations are \(0.905\pm0.003\), \(0.914\pm0.003\), and \(0.938\pm0.002\).
The joint receiver therefore improves the BLR and all-pixel correlations over
edge-first by \(0.123\pm0.006\) and \(0.0328\pm0.0037\), respectively.  The
radial-profile RMSE values decrease from \(0.4133\pm0.0068\) for edge-first to
\(0.3880\pm0.0066\) for the optimal one-copy POVM and
\(0.3055\pm0.0059\) for the joint receiver.

For the paired difference between the joint receiver and its one-copy parent,
a nonparametric percentile bootstrap with \(5\times10^4\) resamples gives
95\% confidence intervals \([0.0708,0.0968]\) for the BLR-annulus correlation
and \([0.0171,0.0321]\) for the all-pixel correlation.  The improvement is
therefore reproducible across paired noise realizations rather than being set
by the representative seed.

The bottom panel of figure~4 in the main text reports the corresponding nuisance-aware
closure-phase SNR gains.  For each loop, the bar is the geometric mean over
ten wavelengths and 36 epochs, while the whiskers show the 5--95\% sample
quantiles.  Across all 3600 wavelength--epoch--loop samples, the geometric-mean
gains relative to edge-first are \(1.202\) for the optimal one-copy POVM and
\(1.921\) for its finite-\(n_s=3\) joint receiver.  The joint receiver gains
an additional factor \(1.598\) over its one-copy parent, with the individual
loop gains ranging from \(1.545\) to \(1.664\).  The asymptotic collective
limit gives a gain of \(3.471\) relative to edge-first. Although the performance of this 3-copy protocol is still distant from the asymptotic quantum limit, it already exhibits clear quantum advantage over single-copy measurement.

All three receivers use the same photon budget, Fourier coverage, paired noise
draws, image constraints, generic regularizers, discrepancy criterion, and
baseline optimization schedule; no translation or morphology-specific core
prior is imposed.  Their comparable reduced residuals show that the differences
in image quality are not produced by unequal RML fitting quality; they arise
primarily from the receiver-dependent data SNR.  This ordering agrees with the
optimizer-independent amplitude--closure SNR gains in
Fig.\,4 of the main text.  The Fisher comparison therefore quantifies the
receiver advantage independently of imaging, while the reconstruction test
shows that the advantage survives a finite-array RML pipeline.

\end{document}